\pdfoutput=1
\documentclass[10pt, doublecolumn]{IEEEtran}
\usepackage{cite}
\usepackage{indentfirst}
\usepackage{graphicx}
\usepackage{changepage}
\usepackage{stfloats}
\usepackage{amsfonts,amssymb}
\usepackage{bm}
\usepackage{algorithm}
\usepackage{algorithmic}
\usepackage{amsmath}
\usepackage{amsmath,amssymb,amsfonts}
\usepackage{times}
\usepackage{url}
\usepackage{cite}
\usepackage{textcomp}
\usepackage{xcolor}
\usepackage{color}
\usepackage{setspace}
\usepackage{enumitem}
\usepackage{float}
\usepackage{diagbox}
\usepackage[T1]{fontenc}
\usepackage[utf8]{inputenc}
\usepackage{authblk}
\usepackage{threeparttable}
\usepackage{booktabs}
\usepackage{footnote}
\usepackage{graphicx}
\usepackage{pifont}
\usepackage{subfigure}
\usepackage{mathrsfs}
\usepackage{bbm}
\usepackage{dsfont}

\allowdisplaybreaks[4]

\newenvironment{proof}{{\indent  \indent \it Proof:}}{\hfill $\blacksquare$}

\begin{document}
\title{Network-Level Integrated Sensing and Communication: Interference Management and BS Coordination Using Stochastic Geometry}

\author{
	Kaitao Meng, \textit{Member, IEEE}, Christos Masouros, \textit{Fellow, IEEE}, Guangji Chen, and Fan Liu, \textit{Senior Member, IEEE}
	\thanks{Preliminary versions of this paper were presented at the IEEE WCNC 2024 \cite{meng2024bs} and 14th CSNDSP \cite{meng2024Networked}.}
	\thanks{K. Meng and C. Masouros are with the Department of Electronic and Electrical Engineering, University College London, London, UK (emails: \{kaitao.meng, c.masouros\}@ucl.ac.uk). G. Chen is with the Nanjing University of Science and Technology, Nanjing 210094, China (email: guangjichen@njust.edu.cn). F. Liu is with the National Mobile Communications Research Laboratory, Southeast University, Nanjing 210096, China (f.liu@ieee.org). }
}

\maketitle


\begin{abstract}
In this work, we study integrated sensing and communication (ISAC) networks with the aim of effectively balancing sensing and communication (S\&C) performance at the network level. Focusing on monostatic sensing, the tool of stochastic geometry is exploited to capture the S\&C performance, which facilitates us to illuminate key cooperative dependencies in the ISAC network and optimize key network-level parameters. Based on the derived tractable expression of area spectral efficiency (ASE), we formulate the optimization problem to maximize the network performance from the view point of two joint S\&C metrics. Towards this end, we further jointly optimize the cooperative BS cluster sizes for S\&C and the serving/probing numbers of users/targets to achieve a flexible tradeoff between S\&C at the network level. It is verified that interference nulling can effectively improve the average data rate and radar information rate. Surprisingly, the optimal communication tradeoff for ASE maximization tends to use all spatial resources for multiplexing and diversity gain, without interference nulling. In contrast, for sensing objectives, resource allocation tends to eliminate interference, especially when there are sufficient antenna resources, because inter-cell interference becomes a more dominant factor affecting sensing performance. This work first reveals the insight into spatial resource allocation for ISAC networks. Furthermore, we prove that the ratio of the optimal number of users and the number of transmit antennas is a constant value when the communication performance is optimal. Simulation results demonstrate that the proposed cooperative ISAC scheme achieves a substantial gain in S\&C performance at the network level.
\end{abstract}   

\newtheorem{thm}{\bf Lemma}
\newtheorem{remark}{\bf Remark}
\newtheorem{Pro}{\bf Proposition}
\newtheorem{theorem}{\bf Theorem}
\newtheorem{Assum}{\bf Assumption}
\newtheorem{Cor}{\bf Corollary}

\section{Introduction}

Due to the scarcity of spectrum and the severe interference between separate sensing and communication (S\&C) systems under the shared use of wireless resources \cite{meng2024Networked, Mishra2019Toward}, there is a significantly growing interest in utilizing unified infrastructure/waveforms/networks to provide both S\&C services, which has motivated the recent research on integrated sensing and communication (ISAC) (also referred to as joint communication and sensing (JCAS)) \cite{Zhang2021OverviewSignal, Liu2022SurveyFundamental, Meng2022UAV}. By simultaneously conveying information to the receiver and extracting information from the scattered echoes, the ISAC technique has emerged as a promising solution to greatly improve the spectrum/cost/energy efficiency of the S\&C functionalities \cite{Cui2021Integrating}. Most recently, the international telecommunication union (ITU) approved the ISAC technique as one of the six key usage scenarios of six-generation (6G) networks. In the literature, most existing works on this topic have focused on performance analysis and optimization design at the link/system level, e.g., waveform optimization, echo signal processing, and resource balance for one or several base stations (BSs) \cite{Ouyang2022Performance, Liu2022Integrated, Meng2023Throughput, Liu2023DistributedUnsupervised}. In general, these studies only consider S\&C interference within a single cell, while the inter-cell interference for network-level ISAC is rarely taken into account.

For large-scale dense ISAC networks, inter-cell interference acts as a critical constraint that restricts network performance. Cooperative ISAC schemes, e.g., coordinated beamforming and cooperative resource allocation design, are promising solutions to avoid/reduce the inter-cell interference to extremely improve the S\&C performance of the whole network \cite{Shin2017CoordinatedBeamforming, Tanbourgi2014Tractable, meng2024cooperative}. Moreover, network-level ISAC also provides new degrees of freedom (DoF) to optimize and balance S\&C performance in a more flexible way, e.g., optimizing cooperative BS cluster sizes for S\&C and the average numbers of served/detected users/targets. Nonetheless, it is challenging to quantitatively analyze and optimize the corresponding average network performance for S\&C in ISAC networks, especially with consideration of the influence of channel fading and spatial attributes. This motivates us to develop new ISAC cooperation approaches to reduce the interference of the networks, and then accurately derive the S\&C performance for further network optimization.
 
Stochastic geometry (SG) provides a powerful mathematical tool for the analysis of multi-cell wireless networks, and has been widely used in a variety of settings \cite{Andrews2011TractableApproach, JoHanShin2012Heterogeneous, Chen2018Stochastic}. For instance, \cite{Andrews2011TractableApproach} proposed a general framework for analysis of the signal-to-interference-plus-noise ratio (SINR), the average data rate, and the coverage probability in multi-cell communication networks. Moreover, in \cite{Bjrnson2016Deploying}, a tractable uplink energy efficiency (EE) maximization problem is formulated, where several practical factors are considered, e.g., the BS density, the transmit power, the number of antennas and users per cell, and the pilot reuse factor. In addition to evaluating the performance of communication-only networks, SG may also be adopted for the sensing performance analysis in vehicular radar networks and wireless sensor networks (WSNs) \cite{al2017stochastic, wang2023performance, Roth2019FundamentalImplications}. For example, in \cite{al2017stochastic}, automotive radar interference is modeled under a two-lane scenario, where the success probability of automotive radars in detecting their targets is estimated. Furthermore, authors in \cite{wang2023performance} proposed a time-frequency division multiple access scheme to mitigate the interference in a coordinated manner and reduce interference probability for mmWave automotive radar networks. In \cite{Roth2019FundamentalImplications}, the authors investigated the localization performance without considering the inter-cell interference, and the results illustrated the impact of increasing the BS density on distance-based localization accuracy in WSNs.

In addition to the performance analysis of the communication-only and sensing-only networks, there are recent works that analyze the ISAC performance by exploiting the SG technique. For instance, \cite{chen2022isac} studied an ISAC beam alignment approach for THz networks, where the reference signal and the synchronization signal block are jointly designed to improve beam alignment performance, and then SG is utilized to derive the expression of coverage probability and network throughput. In \cite{Ram2022Frontiers}, the BS serves as a dual functional transmitter that supports both S\&C functionalities in a time division manner, where during the search interval, the radar scans the entire angular search space to find the maximum number of mobile users. In \cite{salem2022rethinking}, the EE for dual-functional radar-communication (DFRC) cellular networks was derived in closed form to facilitate the  network EE maximization by optimizing the BS density. Furthermore, \cite{Olson2022RethinkingCells} developed a mathematical framework for characterizing the S\&C coverage probability and ergodic capacity in a mmWave ISAC network. However, it is noteworthy that the above studies focus on providing services for only one user and one target in each cell within one frame, which definitely does not fully unleash the multiplexing gain offered by the spatial resources. Moreover, these works seldom consider the cooperation between BSs to prevent/suppress inter-cell interference for network performance improvement. 

In general, there are mainly two kinds of strategies to handle inter-cell interference for both S\&C, i.e., interference cancelling and interference nulling \cite{ghatak2020elastic, Zakhour2012BaseStation, Tanbourgi2014Tractable, Feng2019Location}. The basic principle of the former is to estimate the interfering signal and remove it from the received signal \cite{Xie2023Networked}, which generally requires stringent time synchronization requirements and high overhead over the networks. Clock synchronization becomes a more critical issue in multi-BS S\&C scenarios with increasing bandwidths and carrier frequencies \cite{Song2021TargetLocalization}, since even small errors in the time domain can result in significant range errors during signal processing. Moreover, to cancel the interference from other BSs, it requires the channel state information (CSI) and the original data, leading to higher overhead. Interference nulling involves designing the transmit beamforming in the null space of the interference channel to avoid inter-cell interference. Thus, it does not require data sharing between BSs and thus has a lower signalling overhead \cite{Li2015UserCentric}, which is more practical for dense networks to improve the service quality of S\&C. 

Based on the above discussion, in this work, we study a cooperative ISAC scheme to effectively balance S\&C performance at the network level.\footnote{Considering that the monostatic sensing has loose requirements of time synchronization between BSs, we focus on monostatic sensing instead of multi-static sensing in this work. Therefore, sensing cooperation in our work is with the intent of suppressing inter-cell interference.} We employ SG techniques to conduct performance analysis, shedding light on crucial cooperative dependencies within the ISAC network. This analysis offers valuable insights to guide a closer inspection of specific features and emerging trends. By using SG, the S\&C performance of the entire cellular network can be described more fairly despite the presence of an infinite number of random variables (BS locations, fading values). Different from the closest related works serving/probing one communication user/target at each time \cite{olson2022coverage, salem2022rethinking}, in this work, multiple users and targets are served/detected simultaneously, and the number of users/targets per BS are also optimized to improve the whole network performance. 

In the cooperative ISAC networks, the availability of multiple antennas can provide the flexibility to simultaneously realize three kinds of system objectives: serving/probing multiple users/targets at the same time-frequency resource block to achieve a \textit{spatial multiplexing gain}; providing \textit{spatial diversity gain} for communication users; and \textit{nulling inter-cell interference}. Then, an interesting question in this case is how to balance these competing benefits for both S\&C in ISAC networks. To answer this question, we aim for maximizing two proposed network S\&C performance metrics based on the derived tractable expression, by jointly optimizing the cooperative BS cluster sizes for S\&C and the serving/probing numbers of users/targets. This thus strikes an improved  tradeoff between S\&C at the network level. Surprisingly, as we will show later in this paper, the optimal communication tradeoff for the case of area spectral efficiency (ASE) maximization tends to employ all spacial resources towards multiplexing and diversity gain, without interference nulling. By contrast, for the sensing objectives, resource allocation tends to eliminate certain interference especially when the antenna resources are sufficient, since the interference in the sensing operations is more severe compared to the communications. The main contributions of this paper are summarized as follows:
\begin{itemize}[leftmargin=*]
	\item First, we propose a cooperative ISAC network with the exploitation of interference nulling for S\&C by adopting the coordinated beamforming technique. With the help of the proposed cooperative ISAC network model and SG tools, the S\&C performance is described analytically through the ASE expression to reveal the insights of spatial resource allocation for ISAC networks, and it is verified that the interference nulling can effectively improve the average data rate and radar information rate.
	\item Second, we derive the tractable expressions of the ASE for S\&C, based on which, we prove that it is not necessary to perform interference nulling when maximizing the communication ASE under some general parameter setups; in contrast, interference nulling is required when maximizing the sensing ASE. The primary reason for this susceptibility to inter-cell interference during the sensing operation is the weak effective signal caused by the round-trip pathloss of echo signals. Additionally, we demonstrate a fixed relationship between the optimal number of service users and the number of BS transmit antennas for optimal communication performance, which is essential for achieving optimal spatial diversity and multiplexing tradeoffs.
	\item Third, it is first revealed that, unlike communication, the distribution of sensing interference distances exhibits hole regions and thus the analytical approach in terms of characterizing communication performance cannot be applicable to characterize the sensing performance. To tackle this issue, we propose a geometry-based method to mathematically derive a tight tractable expression of radar information rate by removing the probability with respect to areas where inter-cell sensing interference does not exist. 
	\item Finally, we formulate a performance boundary optimization problem for ISAC networks, whose performance is compared with an inner bound to verify that nulling interference in ISAC networks can effectively improve the cooperation gain and achieve a more flexible tradeoff between S\&C. Moreover, in simulations, it is revealed that ISAC networks also exhibit a tradeoff between the average link performance (the average data rate and radar information rate) and the whole network performance (the ASE) for both S\&C.
\end{itemize}

Notation: $B(a,b,c) = \int_0^a t^{(b-1)} (1-t)^{c-1}dt$ is the incomplete Beta function, $\Gamma(a,b) = \int_0^b t^{a-1}e^{-t}dt$ is the incomplete gamma function. Lowercase letters in bold font will denote deterministic vectors. For instance, $X$ and ${\bf{X}}$ denote one-dimensional (scalar) random variable and random vector (containing more than one element), respectively. Similarly, $x$ and ${\bf{x}}$ denote scalar and vector of deterministic values, respectively. ${\rm{E}}_{x}[\cdot]$ represents statistical expectation over the distribution of $x$, and $[\cdot]$ represents a variable set. ${\cal{O}}(0,r)$ denotes the circle region with center at the origin and radius $r$.

\section{System Model}

\begin{figure}[t]
	\centering
	\includegraphics[width=8.4cm]{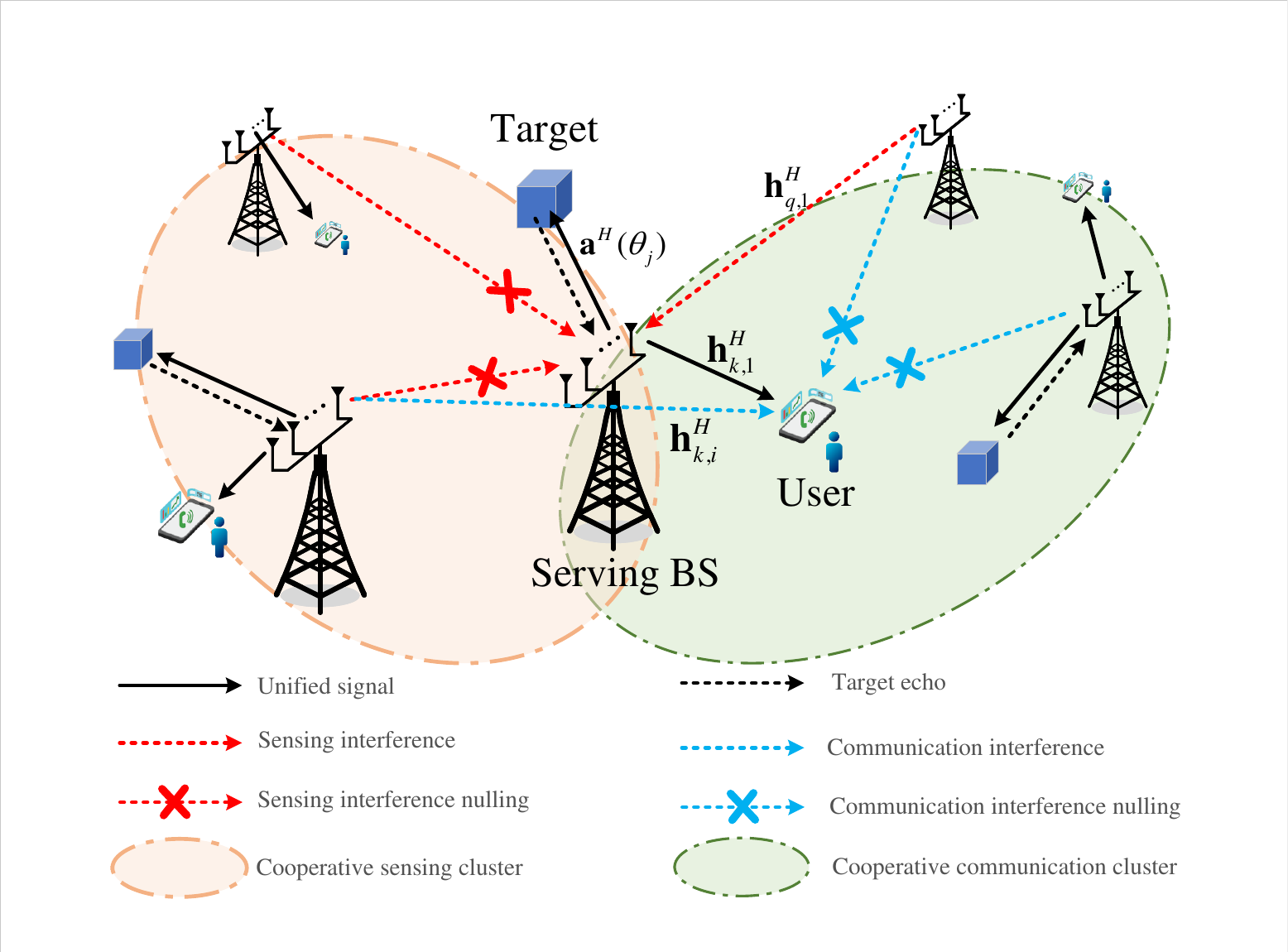}
	\vspace{-3mm}
	\caption{Illustration of cooperative ISAC networks with separate interference nulling for S\&C.}
	\label{figure1}
\end{figure}

\begin{table}[t]
	\footnotesize
	\caption{Important notations and symbols.} 
	\label{Notation}	
	\centering
	\begin{tabular}{ll}%
		\hline
		$\rm{\textbf{Notation}}$  & $\rm{\textbf{ Physical meaning}}$ \\
		\hline
		$\lambda_b, \lambda_u, \lambda_s$   &  Densities of BSs, users, and targets \\
		\hline  
		$Q$                       &  Cooperative cluster size for sensing  \\ 
		\hline  
		$L$                       &  Cooperative cluster size for communication \\
		\hline  
		$K$                       &  Number of served users in each cell \\
		\hline  
		$J$                       &  Number of detected targets in each cell \\
		\hline  
		${\bf{d}}_i$              &  Location of BS $i$ \\
		\hline  
		${\bf{W}}_i$              &  Beamforming matrix of BS $i$ \\
		\hline  
		$M_{\rm{t}}$, $M_{\rm{r}}$ &  Numbers of transmit and receive antennas of each BS \\
		\hline  
		$\alpha$                   &  The pathloss factor of communication channel\\
		\hline  
		$R_c$, $R_s$               &  Average data rate and average radar rate \\
		\hline  
		$T^{\rm{ASE}}_c$, $T^{\rm{ASE}}_s$    &  Area spectral efficiency of communication and sensing \\
		\hline  	
	\end{tabular}
\end{table}

\subsection{Network Model}
In this work, we design a coordinated beamforming scheme to achieve interference nulling in optimized cooperative BS clusters in ISAC networks. Specifically, in the considered system, each BS is equipped with $M_{\mathrm{t}}$ transmit antennas and $M_{\mathrm{r}}$ received antennas, and the location of BSs follows a homogeneous Poisson point process (PPP), denoted by $\Phi_b$. Here, $\Phi_b = \{ {\bf{d}}_i \in \mathbb{R}^2, \forall i \in \mathbb{N}^+\}$, and $ {\bf{d}}_i$ represents the location of BS $i$. Similarly, $\Phi_u$ and $\Phi_t$ respectively represent the point processes of the two-dimensional (2D) locations of single-antenna communication users and targets in the ISAC network. It is assumed that $\Phi_b, \Phi_u$, and $\Phi_t$ are mutually independent PPPs with intensities $\lambda_b$, $\lambda_u$, and $\lambda_s$, where $\lambda_u, \lambda_s \gg \lambda_b$ \cite{Marco2016Stochastic}. 

Focusing on one time-frequency resource block, we consider that each BS sends independent data to $K$ users while simultaneously sensing $J$ targets (e.g., localization and recognition) with unified ISAC signals in the Voronoi cell of the serving BS \cite{Liu2022Integrated}.\footnote{It is assumed that 
the number of users/targets in a network is far larger than the number of deployed BSs, which is generally valid in practical scenarios \cite{Lee2015Spectral}.} We note that compared with the power of the received echo signals which suffer from round-trip path loss, the sensing interference signals are relatively more significant since they are directly transmitted by the surrounding BSs. Therefore, in ISAC networks, sensing operations is more susceptible to inter-cell interference than communications. To suppress S\&C interference separately, we extend the coordinated beamforming model in pure communication network into ISAC networks \cite{Lee2015Spectral, Xu2019Modeling}, where $Q$ ($L$) closest BSs are selected for interference nulling of sensing (communication). Regarding information sharing among intra-cluster BSs, only CSI needs to be shared among BSs for cooperative beamforming, which incurs significantly less overhead compared to sharing data or communication information.

If the interference nulling request number received by the BS exceeds the limitation of DoF of spatial resources, some interference nulling requests will be randomly abandoned. The simulations demonstrate that the impact on network performance of abandoning certain user requests is negligible. In the following, we analyze the network performance based on the average number of interference nulling requests received by each BS to avoid exceeding its beamforming capability. By forming separate S\&C cooperative clusters based on targets' and users' locations, the corresponding beamforming vectors of $Q$ BSs ($L$ BSs) are jointly designed for interference nulling of sensing (communication), as shown in Fig.~\ref{figure1}. Therefore, the interference to the serving BS in the corresponding cooperative cluster of S\&C can be effectively avoided. To fully exploit the integration gain between S\&C in network-level ISAC, in this work, the number of the users and targets to be served, as well as the size of cooperative BS clusters for S\&C are jointly optimized based on the derived tractable expression of the S\&C performance metrics. 

To facilitate the performance analysis, the average number of nulling requests is adopted in this work. Within a given region  ${\cal{A}}$ with area size $|{\cal{A}}|$, the average number of BSs within this area can be denoted by $\bar B = |{\cal{A}}| \lambda_b$ according to the BS density $\lambda_b$ of the PPP. The total communication (and sensing) interference nulling request number received by the BSs within the area $|{\cal{A}}|$ can be expressed as $\bar B K (L-1)$ (and $\bar B J (Q-1)$). In this case, it can be readily proved that there are $\frac{\bar B K (L-1)  + \bar B J (Q-1)}{\bar B} = K(L-1) + J (Q-1)$ service requests received by each BS in this area. Important notations and symbols used in this work are given in Table \ref{Notation}. 

\subsection{Communication model}
Without loss of generality and in the common stochastic-geometry fashion, the typical user is denoted by user $k$ and is assumed to be located at the origin, and its performance is analyzed to generally represent the average performance of all users \cite{Bjrnson2016Deploying}. The typical user is served by its closest BS (namely serving BS), whose index is denoted by 1. The large-scale pathloss fading from the user to the serving BS is considered as $\left\|\mathbf{d}_1\right\|^{-\alpha}$, where $\mathbf{d}_1$ and $\left\|\mathbf{d}_1\right\|$ respectively represents the location of the serving BS and its distance to the typical user, and $\alpha \ge 2$ is the pathloss exponent. Then, the received signal at the typical user $k$ can be expressed as\footnote{The interference model can be easily extended to the case in which different pathloss laws for LOS and NLOS path are adopted \cite{Bai2015Coverage}. The coverage range of direct communication signals affects the decision on the cluster size, and the corresponding analysis can be extended by introducing a spherical spread range.}
\begin{equation}
	\begin{aligned}
		y_{c,k}=& \underbrace{\left\|\mathbf{d}_1\right\|^{-\frac{\alpha}{2}} \mathbf{h}_{k, 1}^H \mathbf{W}_1 \mathbf{s}_1}_{\text{intended signal}} +\underbrace{\sum\nolimits_{i=2}^{L}\left\|\mathbf{d}_i\right\|^{-\frac{\alpha}{2}} \mathbf{h}_{k, i}^H \mathbf{W}_i \mathbf{s}_i}_{\text{intra-cluster interference}}\\
		&+\underbrace{\sum\nolimits_{i=L+1}^{\infty}\left\|\mathbf{d}_i\right\|^{-\frac{\alpha}{2}} \mathbf{h}_{k, i}^H \mathbf{W}_i \mathbf{s}_i}_{\text{inter-cluster interference}} + \underbrace{n_{k,c}}_{\text{noise}},
	\end{aligned}
\end{equation}
where $\mathbf{h}^H_{k, i} \sim \mathcal{C N}\left(0, \mathbf{I}_{M_{\mathrm{t}}}\right)$ is the channel vector from BS at $\mathbf{d}_i$ to user $k$, $\mathbf{W}_i=\left[\mathbf{w}_1^i, \ldots, \mathbf{w}_K^i\right] \in C^{M_{\mathrm{t}} \times K}$ is the  precoding matrix of the BS at $\mathbf{d}_i$, and $\mathbf{s}_i=\left[s_1^i, \ldots, s_K^i\right]^T$ is the information symbol vector transmitted from the corresponding BS. We assume $\mathrm{E}\left[\mathbf{s}_i \mathbf{s}_i^H\right]=\frac{P_{\mathrm{t}}}{K} \mathbf{I}_K$ due to equal power allocation across the BSs, where $P_{\mathrm{t}}$ is the transmission power of BSs.

In light of severe interference within dense cell scenarios, this paper focuses on an interference-limited network where noise can be neglected, utilizing the signal-to-interference ratio (SIR) for performance analysis \cite{Park2016OptimalFeedback, lee2014spectral}. The BS at $\mathbf{d}_i$ uses zero-forcing (ZF) beamforming $\mathbf{W}_i$ with equal power allocated across the $K$ users, while nullifying intra-cluster S\&C interference and maximizing the desired signal strength for the $K$ users in the cluster. Then, the beamforming vector of the serving BS can be designed based on the equation below:
	\begin{equation}\label{TransmitBeamforming}
		\tilde {\bf{W}}_1 = {{\bf{H}}_1}{\left( {\bf{H}}_1 ^H {\bf{H}}_1 \right)^{-1}},
	\end{equation}
	where ${\bf{H}}_1 = [\tilde {\bf{G}}_{{1,c}}, \tilde {\bf{G}}_{{I,c}}, \tilde {\bf{G}}_{2,s},\cdots,\tilde {\bf{G}}_{\tilde Q,s}]$. Here, intra-cell interference channel is $\tilde {\bf{G}}_{{1,c}} = [( {\mathbf{h}}^H_{1, 1})^T,\cdots,( {\mathbf{h}}^H_{K, 1})^T]$, and $\tilde {\bf{G}}_{{I,c}} = [(\tilde {\mathbf{h}}^H_{1, 1})^T,\cdots,(\tilde {\mathbf{h}}^H_{N_I, 1})^T]$ represents the inter-cell interference channel between the serving BS and the users sending interference nulling requests, where $N_I$ represents the number of interference nulling requests received and accepted by the serving BS, with $N_I \le K(L-1)$. 
 Moreover, the intra-cell sensing interference channel is denoted by $\tilde {\bf{G}}_{q,s} = [(\tilde {\bf{h}}_{1,q}^H(\theta_1))^T,\cdots,(\tilde {\bf{h}}_{1,q}^H(\theta_J))^T]$, with $q \in \{2,\cdots, \tilde Q\}$, where $\tilde {\bf{h}}_{1,q}^H(\theta_j)$ denotes the interference channel between the serving BS and target $j$ detected by BS $q$. Similarly, $\tilde Q - 1$ represents the number of BSs sending the sensing interference nulling requests that are accepted by the serving BS, with $\tilde Q \le Q$.
 Therefore, the spatial resource constraint always holds, i.e., $K + N_I + J(\tilde Q-1) \le KL +J(Q-1) \le M_{\mathrm{t}}$. Let $\tilde  {\bf{W}}_i = [\tilde{\bf{w}}^i_1, \cdots, \tilde{\bf{w}}^i_K, \cdots, \tilde{\bf{w}}^i_{N_I+J(\tilde Q-1)}]$, where the first $K$ vectors $\mathbf{W}_i = [{\bf{w}}^i_1, \cdots, {\bf{w}}^i_K] \in {\mathbb{C}}^{M_{\mathrm{t}} \times K}$ form the beamforming matrix of BS $i$, where ${\bf{w}}^i_k=\tilde{\bf{w}}^i_k/\|\tilde{\bf{w}}^i_k\|$. The beamforming design assumes perfect channel knowledge.
Then, the SIR is given as follows: 
\begin{equation}\label{CommunicationSIR}
	{\rm{SIR}}_c = \frac{\frac{P_{\rm{t}}}{K}{g_{k,1}^k{{\left\| {{{\bf{d}}_1}} \right\|}^{ - \alpha }}}}{{\sum\nolimits_{i = L+1}^\infty  \frac{P_{\rm{t}}}{K} {{g_{k,i}}} {{\left\| {{{\bf{d}}_i}} \right\|}^{ - \alpha }}}},
\end{equation}
where $g_{k,1}^k=\left|\mathbf{h}_{k, 1}^H \mathbf{w}_k^1\right|^2$ denotes the effective desired signals' channel gain, $\sum\nolimits_{i=L+1}^{\infty} g_{k,i}\left\|\mathbf{d}_i\right\|^{-\alpha}$ is the remaining inter-cluster  interference, and $g_{k,i} = \sum\nolimits_{j=1}^{K} \left|\mathbf{h}_{k, i}^H \mathbf{w}_j^i\right|^2$.
In this work, ASE \cite{Ding2015AreaSpectralEfficiency, chen2020performance} is considered as the network-level communication performance metric which reflects the network throughput. Then, the mathematical expression for the communication ASE is given by
\begin{equation}\label{ASEcommunication}
	T^{\rm{ASE}}_{c}=\lambda_b K R_c,
\end{equation}
where $R_c=\mathrm{E}[\log (1+\mathrm{SIR}_c)]$ denotes the average data rate of users.

\subsection{Sensing Model}
\label{SensingModelSection}
Similarly, the typical target is denoted by target $j$ and is located at the origin, which is sensed by the nearest BS located at ${\bf{d}}_1$.\footnote{For ease notation, we use the same BS index notation to represent the serving BS's location of S\&C.}  The large-scale pathloss fading from the target to the serving BS is considered as $\left\|\mathbf{d}_1\right\|^{-2\beta}$, where $\beta \ge 2$ is the pathloss exponent from the serving BS to the typical target. Each BS designs the corresponding receive filter towards the directions of $J$ targets. To make the analysis tractable, the maximum-ratio combining (MRC) receive filter ${\bf{v}}_j^H (\theta_j) = \frac{1}{M_{\rm{r}}}[1, \cdots, e^{ -{j \pi(M_{\mathrm{r}}-1)  \cos(\theta_j) }}]^T$ is adopted in this work since MRC offers a reasonable balance between performance and analytic tractability, where $\theta_j$ denotes the direction of target $j$.

The target sensing can be implemented over a coherent processing interval consisting of $N$ snapshots in a monostatic setup. Let ${\bf{S}}_i = [\mathbf{s}_i[1],\cdots,\mathbf{s}_i[N]]$ and $\mathbf{s}_i[n]=\left[s_1^i[n], \ldots, s_K^i[n]\right]^T$, where $i$ denotes the index of transmit BS. The echo signals at the ISAC receiver are denoted by
\begin{align}\label{TargetEchoSignal}
			{\bf{y}}_s  = & {\bf{v}}_j^H(\theta_j)\underbrace { {\xi}{{{\left\| {{{\bf{d}}_1}} \right\|}^{-\beta}}} {\bf{b}}(\theta_j ){{\bf{a}}^H}(\theta_j )}_{{\text{target round-trip channel}}} {\bf{X}}_1 {\bf{P}}_1 \nonumber \\
			&+ \underbrace {\sum\nolimits_{q = 2}^Q { {\left\| {{{\bf{d}}_q} - {{\bf{d}}_1}} \right\|^{-\frac{\alpha}{2}}} {\bf{h}}_{q,1}^H {\bf{X}}_q {\bf{P}}_q }}_{{\text{intra-cluster interference}}} \nonumber \\
			&+ \underbrace{ \sum\nolimits_{{q} = Q+1}^{\infty} { {\left\| {{{\bf{d}}_q} - {{\bf{d}}_1}} \right\|}^{-\frac{\alpha}{2}}  {{\bf{h}}_{q,1}^H} {\bf{X}}_q {\bf{P}}_q} }_{\text{inter-cluster interference}} \nonumber \\
			& + \underbrace{\sum\nolimits_{x = 1}^X {\xi_{c,x}}  {\left\|  {{{\bf{d}}_{c,x}}} \right\|}^{ - \alpha }{\bf{X}}_1 {\bf{P}}_{1,x}}_{\text{clutter interference}} \nonumber \\
			&+ \underbrace {{\bf{v}}_j^H{{\bf{H}}_1} {\bf{X}}_1 {\bf{P}}_1}_{{\text{self-interference}}} + \underbrace {{\bf{v}}_j^H {\bf{N}}_{s}}_{\text{{{noise}}}},
\end{align}
where ${\bf{X}}_q = \mathbf{W}_q\mathbf{S}_q \in {\mathbb{C}}^{M_{\rm{t}} \times N}$ is the signal matrix, ${\bf{y}}_s \in {\mathbb{C}}^{1 \times N}$ denotes the echoes at the BS sensing receiver, ${\bf{P}}_1 = {\rm{diag}}(u(t-2\tau_{j,1}),\cdots,u(t-2\tau_{j,N}))$, ${\bf{P}}_q = {\rm{diag}}(u(t-2\tau_{q,1}),\cdots,u(t-2\tau_{q,N}))$, ${\tau_{j,n}}$ and ${\tau _{q,n}}$ respectively represent the transmission delay of target-serving BS link and BS $q$-serving BS link at the $n$th symbol. Furthermore, ${{\bf{d}}_{c,x}}$ denotes the clutter location, ${\xi_{c,x}}$ represents the radar cross-section (RCS) of the corresponding clutter reflection link, ${\bf{P}}_{1,x}$ represents the transmission delay of clutter reflection link, and $\xi$ denotes the RCS of the target, which can be estimated based on the prior information.\footnote{A target can be generally deemed as a distinct category of user, so it is imperative that the network is capable of discovering targets and, once identified, that their utility within the network is quantifiable and well-defined.} Here, $u(\cdot)$ represents the shifted unit step function. In (\ref{TargetEchoSignal}),
${{\bf{a}}^H}(\theta_j ) = [1, \cdots, e^{ {j \pi(M_{\mathrm{t}}-1)  \cos(\theta_j) }}]$, ${\bf{b}}(\theta_j ) = [1, \cdots, e^{ {j \pi(M_{\mathrm{r}}-1)  \cos(\theta_j) }}]^T$, $\left\| {{{\bf{d}}_q} - {{\bf{d}}_1}} \right\|$ represents the distance from the interfering BSs at ${{\bf{d}}_q}$ to the serving BS, and ${{\bf{H}}_1}$ denotes the self-interference channel from the transmit antennas to the receive antennas of the serving BS. $\mathbf{h}_{q,1}^H$ represents the equivalent channel between BS $q$ and the serving BS, encompassing both the direct interference and clutter-related interference links.

The transmit beamforming of BS $q$ in the cooperative cluster is designed based on the interference channel from BS $q$ to the serving BS (denoted by ${{\bf{G}}_{q,1}^H}$) and the receive filter ${\bf{v}}_j^H (\theta_j)$. By doing so, we only need to eliminate $J \times (Q-1)$ interference channels. For notational simplicity, let ${\bf{h}}_{q,1}^H(\theta_j) = {\bf{v}}_j^H (\theta_j) {\bf{G}}_{q,1}^H$. The receiving filtering yields a performance gain of $\kappa M_{\mathrm{r}}$, with $\kappa \in [0,1]$ denoting the mismatch loss. The value of $\kappa$ can be derived based on the characteristics of Fejer kernel. Then, with a matching filter over the symbol domain, the SIR of echo signals reflected from the typical target $j$ can be given by\footnote{The cancellation of sensing self-interference is based on the known channel ${{\bf{H}}_1}$ and transmit data $\mathbf{s}_1(t)$. Similarly, the clutter interference can be effectively reduced or eliminated based on previous observations since the signals are known at the BS receiver. We also provide an extension method for more detailed consideration of clutter interference in Section \ref{SensingPerformance}.
}
\begin{equation}\label{SensingSIR}
	{\rm{SIR}}_s =  \Delta T \kappa M_{\mathrm{r}} |\xi|^2 \frac{{h_{j,1}^t{{\left\| {{{\bf{d}}_1}} \right\|}^{ - 2\beta }}}}{{\sum\nolimits_{q = Q+1}^\infty  {{h_{q,1}}} {{\left\| {{{\bf{d}}_q} - {{\bf{d}}_1}} \right\|}^{ - \alpha }}}},
\end{equation}
where $h_{j, 1}^t=\sum\nolimits_{k = 1}^K \left|{{\bf{a}}^H}(\theta_j ) \mathbf{w}^1_k\right|^2$ is the effective signal channel gain from the serving BS towards the target's direction, ${\sum\nolimits_{q = Q+1}^\infty  {{h_{q,1}}} {{\left\| {{{\bf{d}}_q} - {{\bf{d}}_1}} \right\|}^{ - \alpha }}}$ is the inter-cluster interference, and ${{h_{q,1}}} = \sum\nolimits_{k = 1}^K \left|{\bf{h}}_{q,1}^H(\theta_j) {\mathbf{w}}^q_k\right|^2$. In (\ref{SensingSIR}), $\Delta T$ denotes the matching filter gain.
Here, for the traditional measurement algorithms (MUSIC, Capon \cite{Stoica1990Maximum}), the maximum distinguishable target's number $J_{\max}$ is limited by the number of receive antennas and the processing time requirements. 

In the literature, it was noticed that the greater the information rate between the target impulse response and the measurement, the better the capability of radar to estimate the parameters describing the target \cite{Yang2007MIMO, Tang2019Spectrally}, and it was proved that the distortion rate at the radar information rate lower bounds the average distortion \cite{Liu2023Deterministic, Dong2023Rethinking}. Hence, radar information rate is an appropriate metric to characterize the estimation accuracy of the system parameters. Several works proved this fact, maximizing the mutual information between the random target ensemble and the target reflections improves the performance in \cite{Yang2007MIMO}. The network-level sensing ASE can be naturally determined through the sensing rate. Here, we estimate ${\bm{\zeta}} \in {\mathbb{R}}^{Z}$ from the observation ${\bf{Y}}_s$ at the sensing receiver. In particular, ${\bm{\zeta}}$ represents the target parameter vector, e.g., angle, range, and velocity, having $Z$ dimension of the target parameters. Let ${\bf{H}}_s =\xi \cdot {\bf{b}}(\theta) \cdot {\bf{a}}(\theta)^H$ denote the target channel. In general, ${\bm{\zeta}} \to {\bf{H}}_s \to {\bf{Y}}_s$ forms a Markov chain \cite{Liu2023Deterministic}, and thus the mutual information between ${\bf{Y}}_s$ and ${\bm{\zeta}}$ equals to that between ${\bf{Y}}_s$ and ${\bf{H}}_s$. Then, we define the sensing ASE to jointly describe the network-level performance of ISAC. Please refer to Appendix A for a detailed derivation of the sensing ASE expression. The mathematical expression for sensing ASE is\footnote{Different from \cite{Olson2022RethinkingCells} dealing with a statistical posterior Cramer-Rao bound (PCRB) for unknown target parameters, our work focuses on the CRB where the system operates with deterministic target parameter settings.} 
\begin{equation}\label{ASEsensing}
	T^{\rm{ASE}}_{s}=\lambda_b J R_s,
\end{equation}
where $R_s=\mathrm{E}[\log (1+\mathrm{SIR}_s)]$ is the target' average radar information rate.

\section{Communication Performance}
\label{CommunicationPerformance}
In this section, we aim to characterize the communication rate analytically via the tools of SG.

\subsection{Expression of Communication Rate}
According to \cite{hamdi2010useful}, for uncorrelated variables $X$ and $Y$, we have
\begin{equation}\label{LemmaEquationLa}
	{\rm{E}}\left[ {\log \left( {1 + \frac{X}{Y}} \right)} \right] = \int_0^\infty  {\frac{1}{z}} \left( {1 - {\rm{E}}\left[{e^{ - z \left[ X\right] }}\right]} \right){\rm{E}}\left[{e^{ - z\left[ Y \right]}}\right]dz.
\end{equation}
Then, under a given distance $r$ from the typical user to the serving BS, the conditional expectation can be derived as follows:
\begin{equation}
	\begin{aligned}
		{\rm{E}}\left[ {\log \left( {1 + \mathrm{SIR}_c} \right)} \big| r \right] & \! = \! {\rm{E}} \! \left[ {\log \left( \! {1 + \frac{{g_{k,1}^k}}{{\sum\nolimits_{{{i}} = L+1}^\infty  {{g_{k,i}}} {{\left\| {{{\bf{d}}_{{i}}}} \right\|}^{ - \alpha }}{r^\alpha }}}} \right)} \right] \\
		&= \int_0^\infty  {\frac{{1 - {\rm{E}}\left[ {{e^{ - zg_{k,1}^k}}} \right]}}{z}} {\rm{E}}\left[ {{e^{ - z I_{\rm{C}}}}} \right]dz,
	\end{aligned}
\end{equation}
where $ I_{\rm{C}} = \sum\nolimits_{{{i}} = L+1}^\infty  {{g_{k,i}}} {{\left\| {{{\bf{d}}_{{i}}}} \right\|}^{ - \alpha }}{r^\alpha }$.
As defined in (\ref{CommunicationSIR}), $g_{k, 1}^k$ is the effective desired signal channel gain, $g_{k, 1}^k \sim \Gamma \left( M_{\mathrm{t}} - KL - J(Q-1) +1,1\right)$ \cite{Hosseini2016Stochastic}, and $g_{k,i} \sim \Gamma \left(K,1\right)$. Based on the above discussion, the useful signal power can be given by 
\begin{equation}\label{UsefulSignalPower}
	\begin{aligned}
		{\rm{E}}\left[ {{e^{ - zg_{k,1}^k}}} \right] &\simeq \int_0^\infty  {\frac{{{e^{ - zx}}{x^{{M_{\mathrm{t}} - KL - J(Q-1) }}}{e^x}}}{{\Gamma \left( {M_{\mathrm{t}} - KL - J(Q-1) + 1} \right)}}} {\rm{d}}x \\
		&= {\left( {1 + z} \right)^{ - ({M_{\mathrm{t}} - KL - J(Q-1) + 1})}},
	\end{aligned}
\end{equation}
where ${M_{\mathrm{t}} - KL - J(Q-1) + 1}$ refers to the diversity gain to each user after consuming a certain DoF to eliminate interference and achieve multiplexing gain. Then, the original expression for $R_c$ is derived in Theorem \ref{CommunicationTightExpression}.

\begin{theorem}\label{CommunicationTightExpression}
The communication performance can be given by
\begin{equation}\label{TightCommunicationExpression}
	\begin{aligned}
		R_c =& \int_0^\infty  {\frac{{1 - {{\left( {1 + z} \right)}^{ - \left( {M_{\mathrm{t}} - LK - J(Q-1) + 1} \right)}}}}{z}} \\
		& \int_0^1 {\frac{{2\left( {L - 1} \right)\eta_L {{\left( {1 - {\eta_L ^2}} \right)}^{L - 2}}}}{ {\rm{H}}\left( {z,K,\alpha ,\eta_L } \right)  + 1 }d\eta_L }dz,
	\end{aligned}
\end{equation}
where ${\rm{H}}\left( {z,K,\alpha ,\eta_L } \right) =K{z^{\frac{2}{\alpha }}}B\left( {\frac{z}{{z + {\eta_L ^{ - \alpha }}}},1 - \frac{2}{\alpha },K + \frac{2}{\alpha }} \right) + \frac{1}{{{\eta_L ^2}}}\left( {\frac{1}{{{{\left( {1 + z{\eta_L ^\alpha }} \right)}^K}}} - 1} \right)$.
\end{theorem}
\begin{proof}
	Please refer to Appendix B. 
\end{proof}

According to (\ref{TightCommunicationExpression}), the average data rate is independent of the BS density $\lambda_b$, and thus the communication ASE $T^{\rm{ASE}}_c$ increases linearly with the BS density. Moreover, it can be observed that the communication performance decreases monotonically as $J(Q-1)$ increases. 
Since $R_c$ in Theorem \ref{CommunicationTightExpression} has complicated relationships with other key system parameters (e.g., $M_{\mathrm{t}}$, $K$, $L$, $J$, and $Q$), we seek for a more tractable expression for $R_c$ for the purpose of exploring the properties between the optimal clustering sizes for S\&C ($Q$/$L$) and the optimal number of targets/users ($J$/$K$). 

\subsection{Approximation of Communication Rate}
To obtain a more tractable expression of the communication rate, we resort to a simple but tight approximation as shown in Appendix C in detail. Before deriving the approximate expression for users’ average data rate, we introduce the mean interference-to-signal-ratio (MISR)-based gain method \cite{Haenggi2014MISR}, which provides a simple and effective way to approximate the performance of a complicated system in comparison with a baseline. Specifically, according to equation (5) in \cite{Haenggi2014MISR}, the coverage probability can be approximated by $P_{L}(\gamma) = P_1(\gamma/G_L)$, where $\gamma$ is the SIR threshold and $G_L$ denotes the effective gain by adopting interference nulling in the cooperative cluster. As a result, the MISR-based method approximates the coverage probability based on the horizontal gap $G_L$ between the SIR distribution to be solved and a reference SIR distribution. In the following, we exploit this approximation in terms of SIR distribution to analyze the average rate and obtain a more tractable expression in Theorem 2.

\begin{theorem}\label{CommunicationLooseExpression}
	The communication rate can be approximated as
	\begin{equation}\label{ApproximationCommunicationRate}
			\tilde R_c = \int_0^\infty  {\frac{{ {1 - {e^{ - z {\frac{{\Gamma \left( {L + \frac{\alpha }{2}} \right)}\left( {{M_{\mathrm{t}}-J(Q-1)+1} - KL} \right)}{K {\Gamma \left( {L + 1} \right)\Gamma \left( {1 + \frac{\alpha }{2}} \right)}}} }}} }}{{z {\rm{F}}(z,\alpha)}}} dz,
	\end{equation}
	where ${\rm{F}}\left( {z,\alpha} \right) = {\left( {{e^{ - z}} - 1} \right) + {z^{\frac{2}{\alpha} }} \Gamma \left( {1 - \frac{2}{\alpha },z} \right)}$.
\end{theorem}
\begin{proof}
	Please refer to Appendix C. 
\end{proof}

Theorem \ref{CommunicationLooseExpression} provides a more tractable form of the data rate, and, based on which, the communication ASE $T^{\rm{ASE}}_c$ can be approximated as
\begin{equation}\label{ApproximateASEcommunication}
T^{\rm{ASE}}_c = \lambda_b K\int_0^\infty  {\frac{{1 - {e^{ - z {Y(K,L,J,Q)} }}}}{{z {\rm{F}}(z,\alpha )}}} dz,
\end{equation}
where $Y(K,L,J,Q) = {\frac{{\Gamma \left( {L + \frac{\alpha }{2}} \right)}\left( {{M_{\mathrm{t}}-J(Q-1)+1} - K L} \right)}{K {\Gamma \left( {L + 1} \right) \Gamma \left( {1 + \frac{\alpha }{2}} \right)}}}$. It is verified that (\ref{ApproximateASEcommunication}) achieves a good approximation by Monte Carlo simulations, as shown in Section \ref{SimulationsSection}. It can be observed that only the term $Y(K,L,J,Q)$ in (\ref{ApproximateASEcommunication}) involves cooperative cluster sizes $L$ and $Q$. In addition, the optimal $L^*$ must belong to the range $[1, \frac{M_{\mathrm{t}}-J(Q-1)+1}{K}]$. Thus, with any given $K$, the optimal cluster size can be obtained by calculating $\frac{\partial Y(K,L,J,Q)}{\partial L} = 0$.

However, it is still difficult to analyze the optimal cooperative cluster size, since the specific relationship of the Gamma function on $L$ is unclear if $\alpha$ takes any value. To draw useful insights, we prove the optimal number of cooperative clusters under some typical parameters in wireless cellular networks.

\begin{Pro}\label{AlphaEqual2}
	When $\alpha \to 2$, the optimal $L^* = 1$ when $T^{\rm{ASE}}_c$ is maximized.
\end{Pro}
\begin{proof}
	When $\alpha \to 2$, $\frac{{\Gamma \left( {L + \frac{\alpha }{2}} \right)}}{{\Gamma \left( {L + 1} \right)\Gamma \left( {1 + \frac{\alpha }{2}} \right)}} \approx 1$. Then, $T^{\rm{ASE}}_c$ decreases monotonically as $L$ increases. Thus, the optimal $L^* = 1$ when the ASE is maximized.
\end{proof}

Proposition \ref{AlphaEqual2} illustrates that when the pathloss coefficient $\alpha$ is close to 2, it is not necessary for multi-antenna BS to use its spatial dimensions to null interference in neighbouring cells. 

\begin{Pro}\label{NoCooperativeNecessary}
	When $\alpha = 4$, the optimal $L^* = 1$ when $T^{\rm{ASE}}_c$ is maximized.
\end{Pro}
\begin{proof}
	Please refer to Appendix D.
\end{proof}

According to Propositions \ref{AlphaEqual2} and \ref{NoCooperativeNecessary}, when $\alpha \to 2$ or $\alpha = 4$, the ASE of communication is generally a decreasing function of the cooperative cluster size $L$. This can be explained by the fact that the communication performance improvement brought by interference nulling cannot compensate for the performance loss caused by the reduction of the diversity/multiplexing gain for the average throughput of the networks. Thus, it is optimal retain all spatial dimensions for multiplexing and diversity. 

Interestingly, contrary to communication, it is verified in simulations that interference nulling is preferred for improving the sensing ASE, as shown in Fig.~\ref{figure7}.

\begin{figure}[t]
	\centering
	\includegraphics[width=7cm]{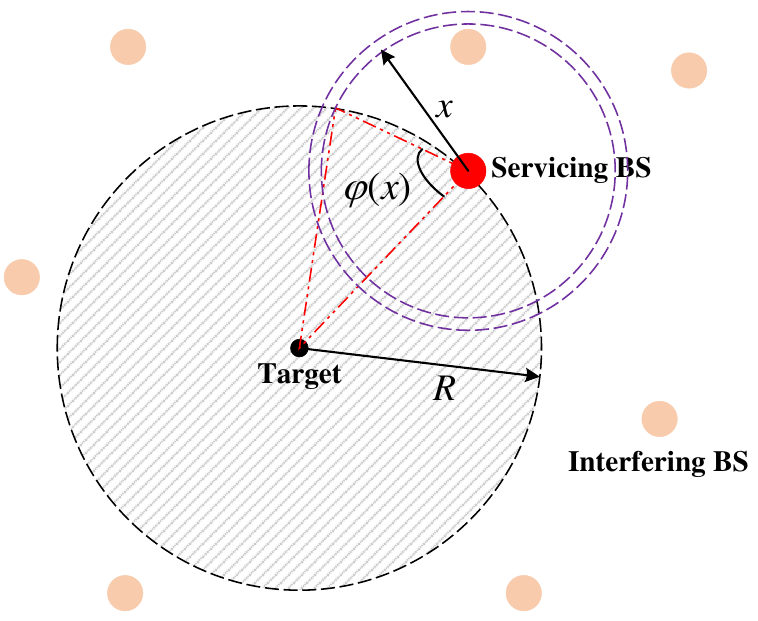}
	\caption{Illustration of sensing interference hole.}
	\label{figure3}
\end{figure}

\section{Sensing Performance}
\label{SensingPerformance}
In this section, we initially note a distinctive hole region in the distribution of sensing interference distances, and then we propose a geometry method to calculate the accurate probability density function (PDF) of interference distances. First, we analyze the distribution of  the effective signal channel gain and the interference channel gain as follows. 
\begin{thm}\label{EffectiveChannelGain}
	The expected transmit beamforming gain at the target's direction and sensing interference channel gain between BSs can be given bt ${\rm{E}}[h^t_{j, 1}]=K$ and ${\rm{E}}[h_{i, 1}]=K$.
\end{thm}
\begin{proof}
	The channel from the interfering BS $q$ to the serving BS is denoted by ${\bf{G}}_{q,1}^H$, where the element of $m$th row and $n$th column of ${\bf{G}}_{q,1}^H$ is $g_{m,n} \sim {\cal{CN}}(0,1)$. The receive filter is ${\bf{v}}_j^H = \frac{1}{\sqrt{M_{\mathrm{r}}}}[1, \cdots, e^{ -{j \pi(M_{\mathrm{r}}-1)  \cos(\theta_j) }}]^T$. After adopting the receive filter ${\bf{v}}_j^H (\theta_j)$, the equivalent interference channel is denoted by ${\bf{h}}_{q,1}^H(\theta_j) = {\bf{v}}_j^H (\theta_j) {\bf{G}}_{q,1}^H $. Let ${\bf{h}}_{q,1}^H(\theta_j) = [h_{1}, \cdots, h_{M}]$, where $h_{m} = \frac{1}{\sqrt{M_{\mathrm{r}}}}(g_{m,1} + e^{ -{j \pi \cos(\theta_j) }} g_{m,2} + \cdots +  e^{ -{j \pi ({M_{\mathrm{r}}}-1) \cos(\theta_j) }} g_{m,{M_{\mathrm{r}}}})$. It can be found that $h_{m}$ is a linear combination of complex Gaussian variables, and thus $h_{m}$ is still a complex Gaussian variable, i.e., ${\bf{h}}_{q,1}^H(\theta_j) \sim {\cal{CN}}(0,{\bm{I}}_{M_{\mathrm{t}}})$. Based on the beamforming matrix in (\ref{TransmitBeamforming}), ${\bf{W}}_1$ is calculated independently of ${{\bf{a}}^H}(\theta_j )$ as the beamforming is designed based on the user channel and inter-cell interfering channel. Therefore, ${\bf{w}}_k$ is independent isotropic unit-norm random vector. Therefore, ${\rm{E}}[h^t_{j, 1}]=K$ and ${\rm{E}}[h_{i, 1}]=K$. This thus completes the proof.
\end{proof}

Even in the presence of angle estimation errors, this Lemma still holds since the beamforming vector is designed based on the random user channels and it can be proved in a similar way. 
For the effective signal channel gain at the target's direction, note that $h^t_{j, 1}$ is an exponentially distributed linear combination of $K$ complex normal random variables. Thus, we utilize $\Gamma(K,1)$ to approximate the effective transmit beamforming gain at the target's direction and sensing interference channel gain based on the expected value of $g_{k,1}^{k}$ and $h_{j,1}^{t}$.
		The distribution of the sensing interference ${{h_{i,1}}}$, as defined below (\ref{SensingSIR}), is a linear combination of $K$ complex normal random variables. Therefore, by neglecting the spatial correlation, $h_{i, 1}$ can be approximated $\Gamma(K,1)$.
Then, under a given distance $R$ from the serving BS to the typical target, we can derive the conditional radar information rate expectation as follows:
\begin{equation}\label{RadarRateExpression}
	\begin{aligned}
		R_s \!=& {\rm{E}}\!\left[ {\log \!\left( \! {1 +  \frac{\Delta T \kappa M_{\mathrm{r}} |\xi|^2  {h_{j,1}^t}}{{\sum\nolimits_{q = Q+1}^\infty  {{h_{q,1}}} {{\left\|  {{{\bf{d}}_q} - {{\bf{d}}_1}} \right\|}^{ - \alpha }} {{R}^{  2 \beta }}}}} \right)}  \! \bigg| \left\|{{\bf{d}}_1}  \right\| \!= \!R \right] \\
		=& \int_0^\infty  {\frac{{1 - {\rm{E}}\left[ {{e^{ - z  \Delta T \kappa M_{\mathrm{r}}|\xi|^2  h_{j,1}^t}}} \right]}}{z}} {\rm{E}}\left[ {{e^{ - z I_{\rm{S}}}}} \right]dz,
	\end{aligned}
\end{equation}
where $I_{\rm{S}} = \sum\nolimits_{q = Q+1}^\infty  {{h_{q,1}}} {{\left\|  {{{\bf{d}}_q} - {{\bf{d}}_1}} \right\|}^{ - \alpha }} {{R}^{  2 \beta }}$.
According to the analysis of the distribution of effective signals and interference signals, we have ${\rm{E}}\left[ {{e^{ - z  \Delta T \kappa M_{\mathrm{r}}|\xi|^2  h_{j,1}^t}}} \right] = {{{\left( {1 + \Delta T \kappa M_{\mathrm{r}} |\xi|^2  z} \right)}^{ - K}}}$. 
We note that it is challenging to derive the radar information rate expression due to the special PDF of distance from the interfering BSs to the serving BS. Specifically, considering that the typical target is sensed by the closest BS, i.e., the serving BS, it is impossible to find another BS in the circle defined with the target as its center and a radius equal to the distance to the serving BS, refer to the gray shaded area in Fig.~\ref{figure3} (namely interference hole).\footnote{The concept of the sensing interference hole was first introduced in \cite{Olson2022RethinkingCells}; this work, however, examines the interference distribution within a multi-BS cooperation framework.} Hence, deriving the radar information rate expression without eliminating the interference hole will result in a deviation from the true value of approximately 15\%, as depicted in Fig.~\ref{figure4}. To the best of our knowledge, no prior work has considered the distribution characteristics of sensing interference in ISAC networks.

To tackle the above issue, we analyze the exact PDF of interfering BS's distance from a perspective of geometry, and obtain a extremely tight expressions with $Q = 1$ as follow. 
\begin{Pro}\label{LaplaceTransformSensing}
	When $Q = 1$, the average radar information rate $R_s$ can be given by
	\begin{equation}\label{SensingRateExpression}
		R_s = \int_0^\infty  {\frac{{1 - {{{\left( {1 +  \Delta T \kappa M_{\mathrm{r}}|\xi|^2 z} \right)}^{ - K}}}}}{z}} \int_0^\infty {{\cal L}_{{I_{\rm{S}}}}}(z) f(R) dR dz,
	\end{equation}
	where
	\begin{equation}\label{LaplaceTransformUnder_R}
		\begin{aligned}
			&{{\cal L}_{{I_{\rm{S}}}}}(z) = \\
			& \!  \exp \bigg(  - R \bigg( K{z^{\frac{2}{\beta }}}{{\left( {\frac{R}{{\pi \lambda_b }}} \right)}^{\frac{{2\alpha }}{\beta } - 1}} \!B\left( {1,1 - \frac{2}{\beta },K + \frac{2}{\beta }} \right) \!+ 1 \\
			&- \! {\int_0^2 \! {\frac{2}{\pi }\arccos  {\frac{t}{2}} \bigg( {1 - {{{{\bigg(\! {1 \!+\! z{{\bigg( {\frac{R}{{\pi \lambda_b }}} \bigg)}^{\alpha  - \beta /2}}\!{t^{ - \beta }}} \bigg)}^{-K}}}}} \! \bigg)} tdt} \bigg) \! \bigg), 
		\end{aligned}
	\end{equation}
	and $f(R) = 2\pi {\lambda _b}R{e^{ - \pi {\lambda _b}{R^2}}}$.
\end{Pro}
\begin{proof}
	Please refer to Appendix E.
\end{proof}

In (\ref{LaplaceTransformUnder_R}), the second part denotes the subtracted interference-free region. BS density generally affects the average sensing performance due to the different pathloss coefficients of effective signals and interference. In the following, we provide a relationship among the pathloss coefficients, the sensing performance and the BS density. Intuitively, when $\alpha > 2\beta$, the sensing performance increases with respect to the BS density. This is due to the fact that as the distance $\left\| {\bf{d}}_1 \right\|$ and $\left\| {\bf{d}}_q - {\bf{d}}_1 \right\|$ decreases, the average effective echo signal power grows more rapidly than that of the interference. On the contrary, when $\alpha < 2 \beta$, the sensing performance decreases with respect to the BS density. When $\alpha = 2 \beta$, we have 
\begin{equation}
	T^{\rm{ASE}}_s = \lambda_b J \int_0^\infty  {\frac{{1 - {{\left( {1 +  \Delta T \kappa M_{\mathrm{r}}z} |\xi|^2  \right)}^{ - K}}}}{z {\rm{I}}(z,K,\alpha)}dz},
\end{equation}
where ${\rm{I}}(z,K,\alpha) = Kz^{\frac{1}{\alpha }}B\left( {1,1 - \frac{1}{\alpha },K + \frac{1}{\alpha }} \right)  - \int_0^2 {\frac{2}{\pi }\arccos \left( {\frac{t}{2}} \right)\left( {1 - \frac{1}{{{{\left( {1 + z{t^{ - 2\alpha }}} \right)}^K}}}} \right)} tdt + 1$. In this case, the sensing ASE $T^{\rm{ASE}}_s$ increases linearly with the BS density. 

However, when $Q \ge 2$, it is challenging to derive a tractable expression of the radar information rate due to the complicated PDF of interfering BSs. To tackle this issue, we exploit the approximation by removing the interference when $Q \ge 2$ to obtain a tractable expression in Theorem 2.

\begin{theorem}\label{ASEsensingExpression}
	When $Q \ge 2$, the sensing ASE can be given by
	\begin{equation}
		T^{\rm{ASE}}_s = \lambda_b J  \int_0^\infty  {\frac{{1 - {{{\left( {1 + \Delta T \kappa M_{\mathrm{r}}z} |\xi|^2 \right)}^{ - K}}}}}{z}} \tilde I_{\mathrm{S}} dz,
	\end{equation}
	where 
	$\tilde I_{\mathrm{S}} \!=\! \int_0^\infty \int_0^\infty \!\exp \bigg(\!  - \!\pi \lambda_b \bigg( r_Q^2\left( {{{{{\left( {1 + z{R^{2\alpha }}r_Q^{ \!- \beta }} \right)}^{-K}}}} \!-\! 1} \!\right) \\ + \! K{z^{\frac{2}{\beta }}}{R^{\frac{{4\alpha }}{\beta }}} \! B \!\bigg(\! {\frac{{z{R^{2\alpha }}r_Q^{ - \beta }}}{{z{R^{2\alpha }}r_Q^{ - \beta } + 1}},1 \!- \frac{2}{\beta },K + \frac{2}{\beta }} \!\bigg) \! \bigg)\! \bigg)\! {f_{{r_{q}}}}\left( r \right)f_R(r) dRdr_q$.
\end{theorem}
\begin{proof}
	Please refer to Appendix F.
\end{proof}

Additionally, when considering clutter interference, it can be analyzed in a similar manner as discussed in Remark \ref{ClutterInterferenceRemark}.
\begin{remark}\label{ClutterInterferenceRemark}
			According to the equation (\ref{LemmaEquationLa}), the conditional expectation can be derived as follows:
			\begin{equation}\label{RadarRateExpression2}
				\begin{aligned}
					R_s = \int_0^\infty  {\frac{{1 - {\rm{E}}\left[ {{e^{ - z \xi \Delta T \kappa M_{\mathrm{r}} h_{j,1}^t}}} \right]}}{z}} {\rm{E}}\left[ {{e^{ - z I_{\rm{S}}}}} \right] {\rm{E}}\left[ {{e^{ - z I_{\rm{X}}}}} \right]dz,
				\end{aligned}
			\end{equation}
			where $I_{\rm{S}} = \sum\nolimits_{q = Q+1}^\infty  {{h_{q,1}}} {{\left\|  {{{\bf{d}}_q} - {{\bf{d}}_1}} \right\|}^{ - \alpha }} {{R}^{  2 \beta }}$ and $I_{\rm{X}} = \sum\nolimits_{x = 1}^X |\xi_{c,x}|^2  {\left\|  {{{\bf{d}}_{c,x}}} \right\|}^{ - 2 \beta }{{R}^{  2 \beta }}$. Here, $I_X$ denotes the clutter interference, and $\xi_x$ represents the RCS of the corresponding clutter reflection link. Then, by deriving the Laplace transformation of $I_X$ in a similar way to Proposition 3 of the revised manuscript, the conditional expectation of radar information rate can be obtained.
\end{remark}

\vspace{-0.5mm}
\section{System Optimization}
In this section, we study the optimization of the cooperative ISAC networks to achieve a flexible balance between S\&C. Based on the derivations in Sections \ref{CommunicationPerformance} and \ref{SensingPerformance}, the ASE for S\&C are both functions of the serving numbers of users $K$ and targets $J$ as well as the cooperative cluster sizes $L$ and $Q$. Then, we present two performance metrics, including the sensing-communication (S-C) ASE region and the sum ASE, to 
verify the effectiveness of the proposed cooperative ISAC schemes. First, we propose to use the S-C ASE region (defined below) to characterize all the achievable communication ASE and achievable sensing ASE pairs under the constraint of the spatial DoF. Without loss of generality, the S-C network performance region is defined as
\begin{equation}
	\begin{aligned}
			\mathcal{C}_{\mathrm{c}-\mathrm{s}}(K,L,J,Q) \triangleq & \big\{ (\hat r_c, \hat r_s): \hat r_c \leq  T^{\rm{ASE}}_c,  \hat r_s \leq  T^{\rm{ASE}}_s,\\
		&   KL +J(Q-1) \le M_{\mathrm{t}}, J \le J_{\max} \big\},
	\end{aligned}
\end{equation}
where $(\hat r_c, \hat r_s)$ represents an achievable S\&C performance pair.
A direct way to find the boundary of the S-C region (as shown in Fig.~\ref{figure2}) is to exhaustively search all feasible variables $(K,L,J,Q)$ and calculate the corresponding S\&C ASE derived in Sections \ref{CommunicationPerformance} and \ref{SensingPerformance}. However, such an operation is with high computational complexity, especially when the number of transmit antennas is large.

It is observed that in (\ref{ApproximateASEcommunication}), to maximize the communication ASE $T^{\rm{ASE}}_c$ efficiently, we can relax $K$ to continuous variables and substituting $v = \frac{K}{ M_{\mathrm{t}}-J(Q-1)+1}$ into (\ref{ApproximationCommunicationRate}), we have 
\begin{equation}
	T^{\rm{ASE}}_c = \lambda_b (M_{\mathrm{t}}-J(Q-1)+1) G(v),
\end{equation}
where $G(v) = \int_0^\infty  {\frac{{1 - {e^{ - z {\frac{{\Gamma \left( {L + \frac{\alpha }{2}} \right)}\left( {\frac{1}{v} - L} \right)}{{\Gamma \left( {L + 1} \right)}\Gamma \left( {1 + \frac{\alpha }{2}} \right)}} }}}}{{zF(z,\alpha )}}} dz$. Therefore,
we know that the optimal $v^*$ is only related to variable $L$. The following lemma gives the optimal $v^*$ to maximize $T^{\rm{ASE}}_c$.

\begin{thm}\label{OptimalKL}
	With any given $J$ and $Q$, the optimal $K$ can be uniquely found by solving the following equation:
	\begin{equation}
		\int_0^{\infty} \frac{1-e^{-z\left(\frac{1}{v}-L\right)}-z e^{-z\left(\frac{1}{v}-L\right)} / v}{z {\rm{F}}(z, \alpha) } d z=0.
	\end{equation}
\end{thm}
\begin{proof}
	$G(v)$ is a concave function with respect to $v$ since the second-order derivative of $G(v)$ satisfies $G''(v)>0$. The fist-order derivative of $G(v)$ is $G'(v) = \int_0^{\infty} \frac{1-e^{-z\left(\frac{1}{v}-L\right)}-z e^{-z\left(\frac{1}{v}-L\right)} / v}{z {\rm{F}}(z, \alpha) } d z$. Then, we have $\lim\limits_{v \to 0} G'(v) = \int_0^{\infty} \frac{1}{z {\rm{F}}(z, \alpha) } d z > 0$, and $G'(1) = \int_0^{\infty} \frac{1-e^{-z\left(1-L\right)}-z e^{-z\left(1-L\right)} }{z {\rm{F}}(z, \alpha) } d z < 0$. Considering that $G(v)$ is a concave function and $G'(v) < 0$, there is always a unique solution $v^*$ within $[0,1]$ to solve the equation $G'(v) = 0$.
\end{proof}

Based on Lemma \ref{OptimalKL}, the optimal $K$ can be obtained by a binary search instead of a one-dimensional search. Moreover, with $L=1$, the optimal ratio of scheduled users $\frac{K}{M_{\mathrm{t}}+1}$ is constant, which is also verified in simulations (c.f. Fig.~\ref{figure5a}).
For the profile of S\&C, we first identify two corner points of this S-C region denoted by $(\hat r_c, r_s^{\max})$ and $(r_c^{\max}, \hat r_s)$, respectively, where $(\hat r_c, r_s^{\max})$ indicates the minimum sensing ASE constrained by the maximum communication ASE, while $(r_c^{\max}, \hat r_s)$ indicates the minimum communication ASE constrained by the maximum sensing ASE, as shown in Fig.~\ref{figure2}. 

\begin{remark}
	According to the Theorems \ref{CommunicationTightExpression} and \ref{CommunicationLooseExpression}, it can be verified that the communication ASE $T^{\rm{ASE}}_c$ decreases monotonically with the increase of $J(Q-1)$. Thus, at the optimal communication point of the S-C region, $J = J_{\max}$ and $Q = 1$, then the optimal $K^*$ and $L^*$ can be obtained by Lemma \ref{OptimalKL}. The optimal communication part is to satisfy $\frac{K}{ M_{\mathrm{t}}-J(Q-1)+1} = v$, where $v$ can be deemed as a constant parameter with a given $L$.
\end{remark}

\begin{thm}
	For the optimal sensing performance, $K^* = 1$ and $L^* = 1$, where the optimal $K^*$ and $L^*$ mean in terms of serving users and cooperative cluster size for communication.
\end{thm}
\begin{proof}
	Similar to the adopted approximation operation in (\ref{ConditionalExpectionC}) in Appendix C,  the SIR expression of radar information rate can be approximated as
	\begin{equation}\label{ApproximateSensing}
		\begin{aligned}
			&{\rm{E}}_{g, \Phi_b}\left[ {\log \left( {1 + \mathrm{SIR}_s} \right)} \big| R\right] \\
			\approx& {\rm{E}}_{\Phi_b}\left[ {\log \left( {1 + \frac{ \Delta T \kappa M_{\mathrm{r}} |\xi|^2  {\rm{E}}\left[h_{j,1}^t\right]}{{\sum\nolimits_{q = Q+1}^\infty  {\rm{E}}\left[ {{h_{q,1}}} \right] {{\left\|  {{{\bf{d}}_q} - {{\bf{d}}_1}} \right\|}^{ - \alpha }} {{R}^{  2 \beta }}}}} \right)} \right] \\
			=& {\rm{E}}_{\Phi_b}\left[ {\log \left( {1 + \frac{ \Delta T \kappa M_{\mathrm{r}}|\xi|^2 }{{\sum\nolimits_{q = Q+1}^\infty   {{\left\|  {{{\bf{d}}_q} - {{\bf{d}}_1}} \right\|}^{ - \alpha }} {{R}^{  2 \beta }}}}} \right)} \right]. 
		\end{aligned}
	\end{equation}
	It is found that the radar information rate is not related to the number of users $K$. Thus, using the minimum number of users and communication cooperative cluster, i.e., $K^* = 1$ and $L^* = 1$, the feasible region of the sensing variables is largest, and thus achieving the optimal sensing performance.
\end{proof}

\begin{figure}[t]
	\centering
	\includegraphics[width=7cm]{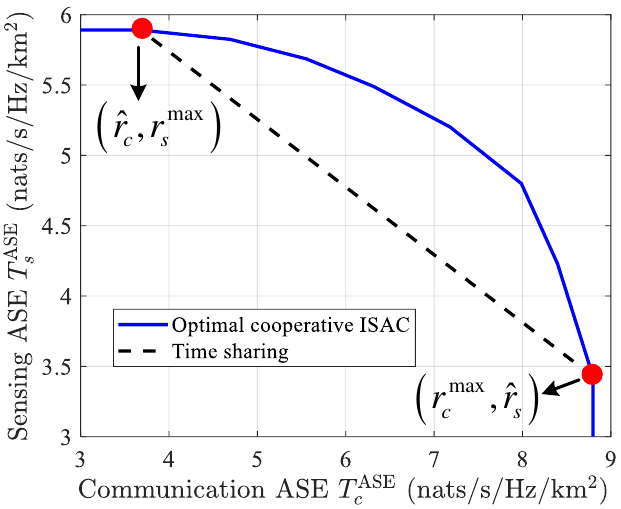}
	\caption{Illustration of the S-C network performance region.}
	\label{figure2}
\end{figure}

Then, for the optimal radar information rate, the optimal $J^*$ and $Q^*$ can be obtained by a two-dimensional search.
Here, an inner bound on the S-C region of Fig.~\ref{figure2} can be achieved with a simple time-sharing strategy based on these two corner points, i.e., $(\hat r_c, r_s^{\max})$ and $(r_c^{\max}, \hat r_s)$.
In addition to these two corner points of the S-C performance region, by exploiting the monotonic increase/decrease of sensing/communication performance with variable $J(Q-1)$, we can obtain the boundary of the S-C region by a binary search over $J(Q-1)$. With each given $J(Q-1)$, the optimal $T^{\rm{ASE}}_c$ and $T^{\rm{ASE}}_s$ can be separately obtained by a two-dimension search. It is worth noting that since at the performance boundary, $Q$ and $J$ must be integers, we also need to verify that there are feasible $Q$ and $J$ with a given integer $J (Q-1)$.

Let us consider the case that there are minimum service quality constraints for the users/targets being served/detected, such as sensing accuracy and real-time communication quality, i.e., $L \ge L^{th}$ and $Q \ge Q^{th}$. In this case, it is not difficult to find that the S\&C ASE is a non-increasing function with respect to the service quality constraints $L^{th}$/$Q^{th}$ since the feasible region becomes smaller with the increase of $L^{th}$/$Q^{th}$. Therefore, there is a new tradeoff between the average performance of S\&C ($R_s$ and $R_c$) and the whole network performance ($T^{\rm{ASE}}_s$ and $T^{\rm{ASE}}_c$).

Next, the second metric of ISAC networks, i.e., the sum ASE, is defined as a function of the S\&C rates and given by
\begin{equation}\label{WeightedSum}
	T^{\rm{ASE}} = \rho T^{\rm{ASE}}_c + \left(1-\rho \right)T^{\rm{ASE}}_s,
\end{equation}
which represents the total ASE of the ISAC network. Similarly, we can optimize the total ASE $T^{\rm{ASE}}$ to illustrate that the spectral efficiency of ISAC networks is effectively improved compared to sensing-only and communication-only networks. In (\ref{WeightedSum}), $\rho$ represents the weighting factor of S\&C performance. The problem formulation can be expressed as
\begin{alignat}{2}
	\label{P1}
	(\rm{P1}): & \begin{array}{*{20}{c}}
		\mathop {\max }\limits_{J,K,Q,L} \quad  T^{\rm{ASE}}
	\end{array} & \\ 
	\mbox{s.t.}\quad
	& KL +J(Q-1) \le M_{\mathrm{t}}, & \tag{\ref{P1}a}\\
	& J,K,Q,L \ge 1, & \tag{\ref{P1}b} \\
	& J \le J_{\max}. & \tag{\ref{P1}c}
\end{alignat}
It is not difficult to prove that the S-C performance boundary is a convex region, then the optimal total ASE can be obtained by searching the ASE of the boundary point.

\begin{figure}[t]
	\centering
	\includegraphics[width=7.5cm]{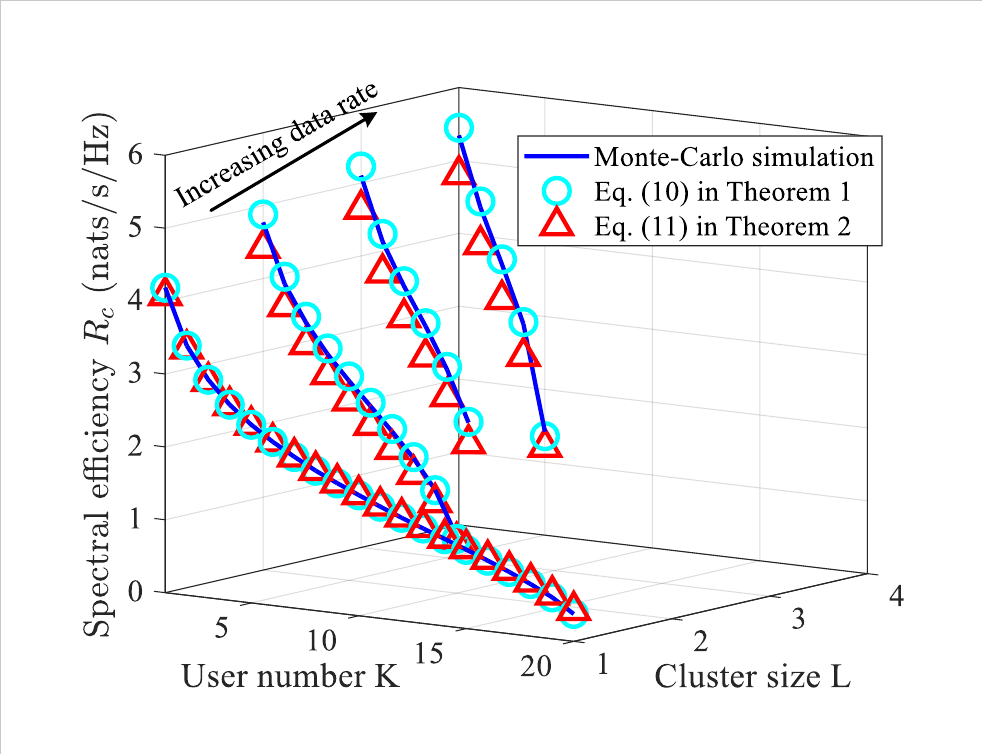}
	\caption{User's average data rate $R_c$ with respect to $K$ under different cooperative cluster size $L$.}
	\label{figure5b}
\end{figure}

\section{Simulations}
\label{SimulationsSection}

Using numerical results, we study the fundamental insights of ISAC networks and verify the tightness of the derived tractable expression by comparing with Monte Carlo simulation results in this section. Numerical simulations are averaged over network typologies and small-scale channel fading realizations. The system parameters are given as follows: the number of transmit antennas $M_{\mathrm{t}} = 20$, the number of receive antennas $M_{\mathrm{r}} = 10$, the transmit power $P_{\mathrm{t}} = 1$W at each BS, $|\xi|^2 = 0.1$, $\kappa = 1$, matching filter gain $\Delta T = 1$, the BS density $\lambda_b = 1/km^2$, $J_{\max} = 10$, pathloss coefficients $\alpha = 4$, and $\beta = 2$.

\subsection{Communication Performance Maximization}

Fig.~\ref{figure5b} illustrates that the results of the original expression for $R_c$ in Theorem 1 are consistent with the simulation results, which validates our analysis and derivation in Section \ref{CommunicationPerformance}. The Monte Carlo simulation results are generated based on the actual number of interference nulling requests received at each BS, where some interference nulling requests may be randomly abandoned due to the limitation of DoF in spatial resources.
Moreover, the gap between the results in Theorem 2 and the simulation results is quite small, which demonstrates that the approximation $\tilde R_c$ in Theorem 2 is quite tight and it becomes more tight for $L = 1$. As shown in Fig.~\ref{figure5b}, under a given $L$, the communication spectral efficiency $R_c$ decreases as $K$ increases. The main reason is that the same spectrum resources of the system are allocated to more users, and thus the average data rate of each user will reduce. Moreover, one can observe that with $K = 1$, as the cooperative cluster size $L$ increases, the average data rate of the scheduled user will increase. The main reason is that if only single user is scheduled in each cell, the performance gain improvement brought by interference nulling is higher than the reduction of the diversity gain. Differently, with $K = 5$, as $L$ increases, the spectral efficiency increases first and then decreases. This suggests that when providing services with more users simultaneously, the benefit of interference nulling does not always outweigh its cost to the reduction of diversity gain, thereby reducing the average data rate $R_c$ for large cluster size. 

\begin{figure}[t]
	\centering
	\includegraphics[width=7.5cm]{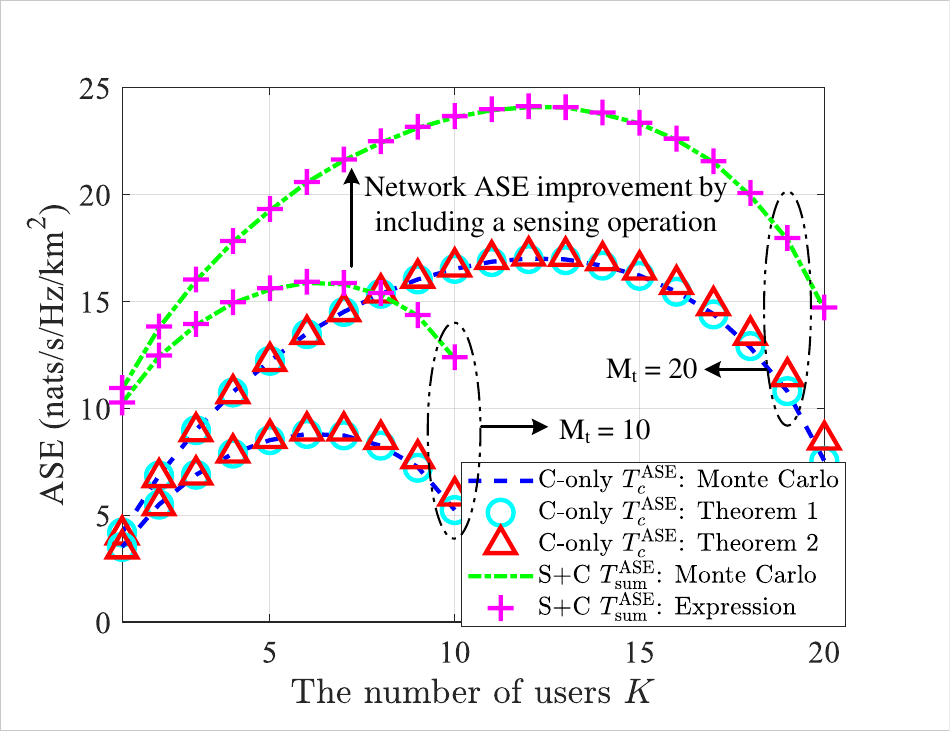}
	\caption{Communication ASE ($T^{\rm{ASE}}_c$) and the total ASE ($T^{\rm{ASE}}_{\rm{sum}}$) with respect to $K$, with $L = 1$.}
	\label{figure5a}
\end{figure}

In Fig.~\ref{figure5a}, with $M_{\mathrm{t}} = 10$ and $M_{\mathrm{t}} = 20$, respectively, the tractable expressions derived based on Theorems  \ref{CommunicationTightExpression} and \ref{CommunicationLooseExpression} provide a remarkably tight approximation which is invariably indistinguishable from truth across all user number cases. The total ASE of the ISAC networks $T^{\rm{ASE}}_{\rm{sum}} = T^{\rm{ASE}}_c + T^{\rm{ASE}}_s$ is significantly higher than the communication ASE $T^{\rm{ASE}}_c$. This is attributed to the fact that adding the analysis function of sensing data to the traditional communication network can effectively improve the overall spectrum efficiency of the network. In addition, Fig.~\ref{figure5a} shows that compared to communication-only networks, the ISAC network achieves approximately constant improvements for different user scheduling numbers, which can be explained by the analysis in (\ref{ApproximateSensing}). In addition, when we maximize the communication ASE, the ratio between the number of users and the number of BS antennas is approximately 60\% under different numbers of transmit antennas, which is consistent with the analysis of Lemma \ref{OptimalKL}. Moreover, it can be found from Fig.~\ref{figure5a} that when the number of scheduled users is less, increasing the number of transmit antennas can only bring limited ASE performance improvement. This is expected since the reasonable spatial resource allocation for multiplexing and diversity gain can effectively maximize the performance gain of ISAC networks. 

To verify the analysis of optimal spatial resource allocation for the maximization of the communication ASE $T^{\rm{ASE}}_c$, we compared $T^{\rm{ASE}}_c$ versus different $L$ and $K$ in Fig. \ref{figure6}. It is shown that at the optimal communication ASE, $K = 12$ and $L = 1$, i.e., there is no need to do interference nulling for the communication ASE maximization, which confirms our analysis in Proposition \ref{NoCooperativeNecessary}. The reason is again that more spatial DoF is consumed to serve interference suppression, and then the multiplexing gain and diversity gain are inevitably reduced, the overall performance of the network decreases.

\begin{figure}[t]
	\centering
	\includegraphics[width=7.5cm]{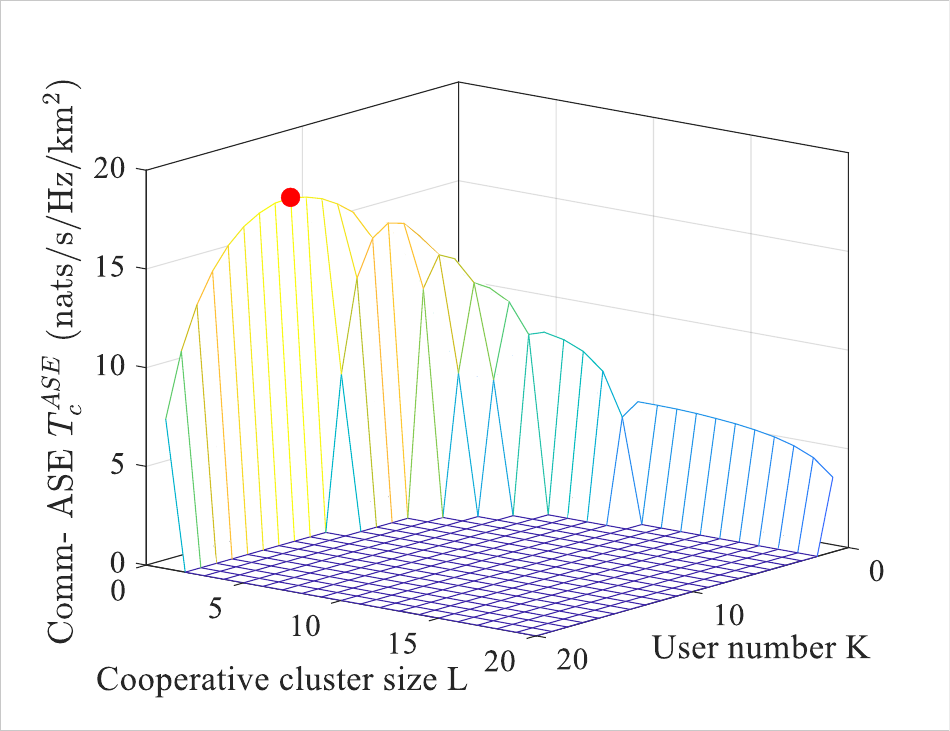}
	\caption{Illustration of the optimal cooperative cluster size $L$ of the maximized ASE of communication.}
	\label{figure6}
\end{figure}

\subsection{Sensing Performance Maximization}

\begin{figure}[t]
	\centering
	\includegraphics[width=7.5cm]{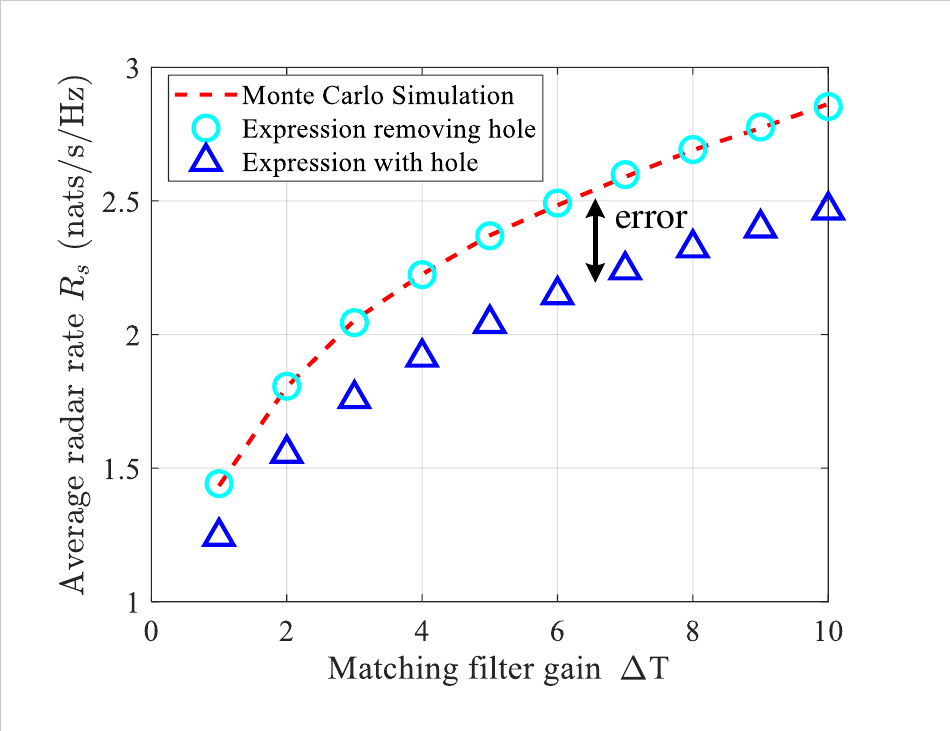}
	\caption{ Comparisons of spectral efficiency $R^{\rm{ASE}}_s$ between the proposed expression and that without consideration of hole.}
	\label{figure4}
\end{figure}

In this subsection, we set $M_{\mathrm{r}} = 10$ and $J_{\max} = 10$, to verify the derived expression under different matching filter gain and cooperative cluster size, as shown in Figs.~\ref{figure4} and \ref{figure7}. Fig.~\ref{figure4} shows that without removing the interference within the hole between the target and the serving BS, the radar information rate of the derived expression is about 15\% less than that of the Monte Carlo simulations. This is expected because the benchmark scheme ignores that there is a circle area without interfering BSs (c.f. the gray shaded circle in Fig.~\ref{figure3}), i.e., more interference is involved, leading to a radar information rate that is underestimated compared to the actual value. Fig.~\ref{figure4} illustrates that by removing the interference in the hole, the derived expression in Proposition \ref{LaplaceTransformSensing} achieves extremely accurate description of radar information rate under different matching filter gain.
\begin{figure}[t]
	\centering
	\includegraphics[width=7.5cm]{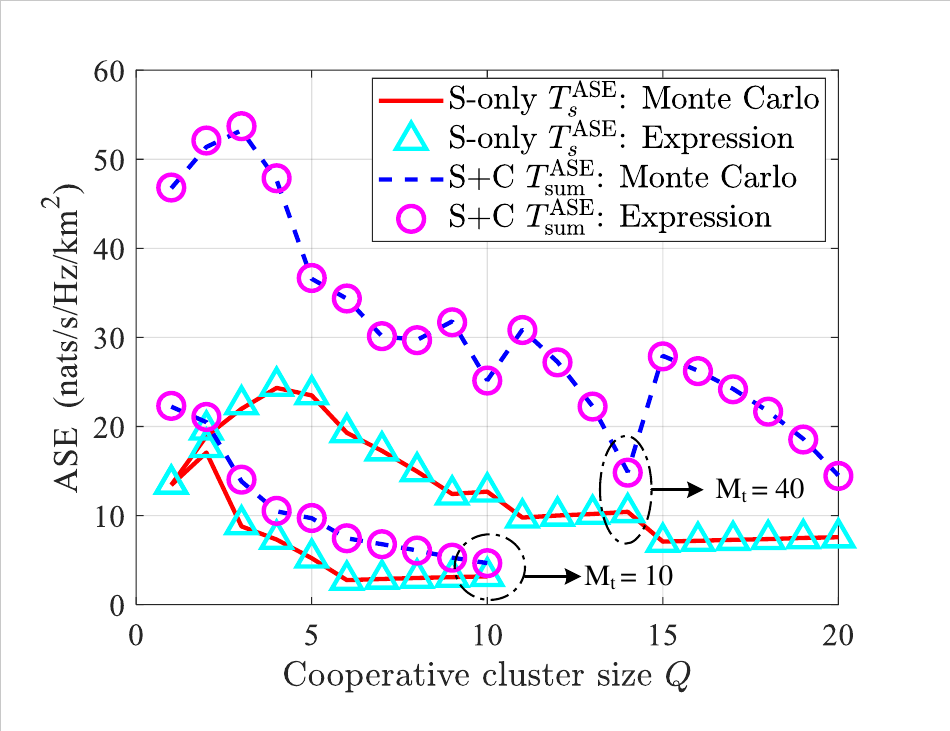}
	\caption{Sensing ASE $T^{\rm{ASE}}_s$ and the total ASE $T^{\rm{ASE}}_{\rm{sum}}$ comparisons with respect to cooperative sensing cluster size $Q$.}
	\label{figure7}
\end{figure}

As shown in Fig. \ref{figure7}, with $K = 1$ and $L = 1$, the tightness of the derived expression of the sensing ASE $T^{\rm{ASE}}_s$ is verified. It is observed that the sensing ASE first increases and then decreases when $Q$ increases. Thus, to maximize the sensing ASE, though the number of targets $J$ that can be sensed simultaneously will inevitably be reduced, interference nulling for sensing with a proper cooperative cluster size can effectively improve the sensing ASE. This is expected since the interference does not suffer from round-trip pathloss as echo signals and that the distance from other interfering BSs may be closer than that from the target. Specifically, when $M_{\mathrm{r}} = 40$, our proposed cooperative scheme can achieve up to twice the ASE of that with no interference nulling case ($Q=1$). With more transmit antennas, there are more DoF for interference nulling to achieve the maximum sensing ASE. It is worth noting that the optimal total ASE performance $T^{\rm{ASE}}_{\rm{sum}}$ does not necessarily increase with $Q$, especially when the number of transmit antennas is relatively small. The main reason is that as $Q$ increases, the communication ASE $T^{\rm{ASE}}_c$ will drop sharply, which may not be compensated by the performance increase in $T^{\rm{ASE}}_s$. Moreover, since $J(Q-1)$ must be an integer, as the cooperative cluster size increases, $J$ may be appropriately reduced to satisfy the DoF constraints (c.f. ({\ref{P1}a})). This causes the sensing ASE to fluctuate as the cooperative cluster size increases, as shown in Fig.~\ref{figure7}. 

\subsection{Tradeoff Between Sensing and Communication}

\begin{figure}[t]
	\centering
	\includegraphics[width=7.5cm]{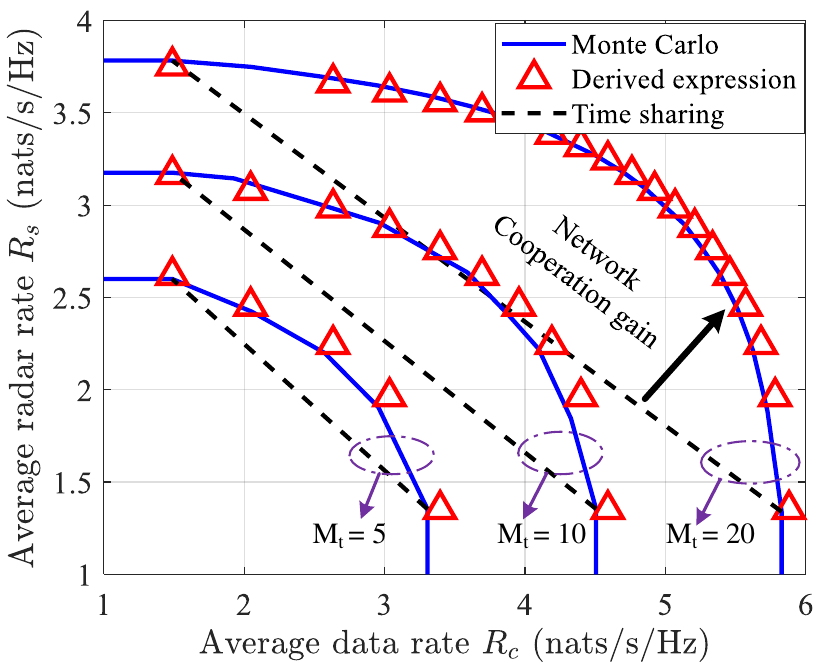}
	\caption{Average user/target's spectral efficiency tradeoff between S\&C.}
	\label{figure9}
\end{figure}

In this subsection, we verify the effectiveness of the proposed cooperative ISAC scheme, including the ASE performance boundary $\mathcal{C}_{\mathrm{C}-\mathrm{S}}(K,L,J,Q)$, with $J_{\max} = 5$. Here, the time-sharing scheme achieved based on two corner points is compared to illustrate the effectiveness of the cooperative ISAC scheme under different setups. The tradeoff profile of the average data rate $R_c$ and the average radar information rate $R_s$ is shown in Fig. \ref{figure9}, which confirms the accuracy of analytical results and the flexibility of our proposed cooperative ISAC networks. With the increase in the number of transmit antennas, the performance boundaries of S\&C are expanded significantly. Also, it is observed from Fig. \ref{figure9} that the $(R_c, R_s)$ region of the optimal cooperative scheme becomes much larger than that of the time-sharing scheme as the number of transmit antennas increases. For example, with $M_{\mathrm{t}} = 10$ and $M_{\mathrm{t}} = 20$, under the same data rate $R_c$, the radar information rates $R_s$ of the proposed cooperative scheme are respectively 38\% and 67\% higher than that of the time-sharing scheme. This expectation arises from the increased number of transmit antennas, enhancing the network's DoF in effectively coordinating spatial resources, leading to improved multiplexing, diversity, and interference nulling gains for S\&C.

\begin{figure}[t]
	\centering
	\includegraphics[width=7.5cm]{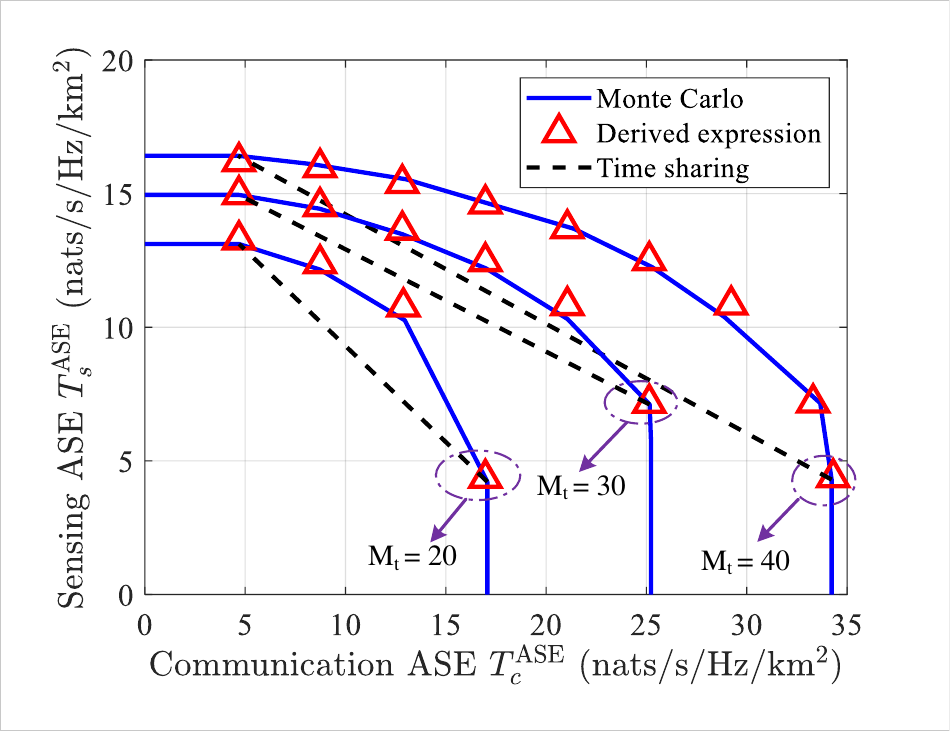}
	\caption{Area spectral efficiency tradeoff between S\&C versus different transmit antenna number $M_t$.}
	\label{figure11}
\end{figure}

\begin{figure}[t]
	\centering
	\includegraphics[width=7.5cm]{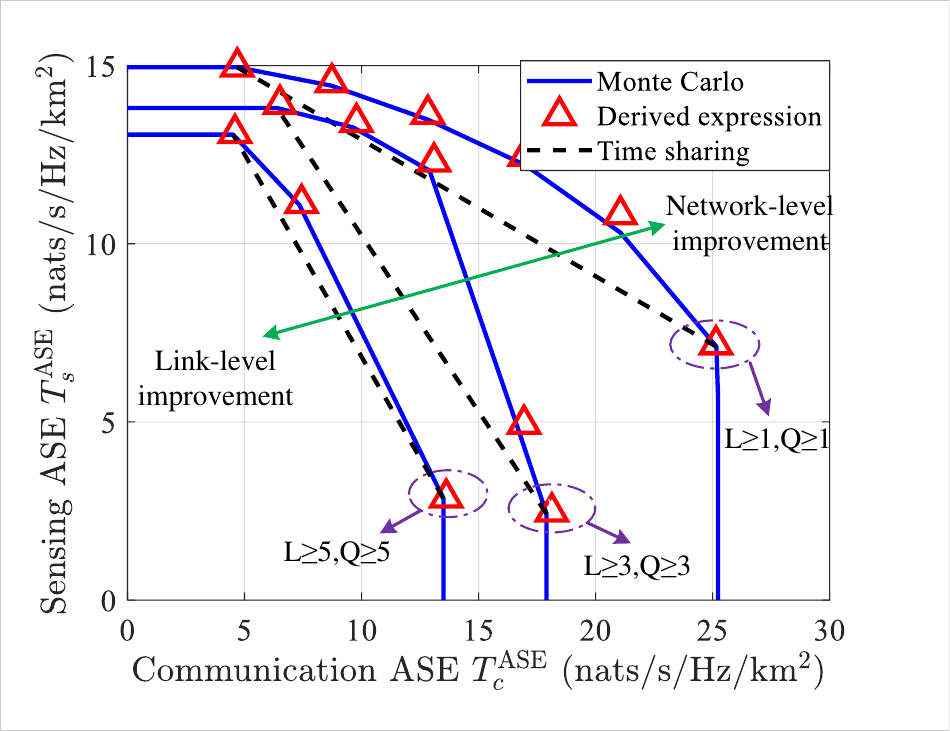}
	\caption{Area spectral efficiency tradeoff between S\&C versus different link-level constraints.}
	\label{figure10}
\end{figure}

Fig.~\ref{figure11} extends the comparison from the individual rates to the network ASEs and illustrates that the network ASE performance can be effectively extended by exploiting the optimal cooperative strategy. Specifically, the proposed cooperative scheme can improve the communication performance by up to 48\% and 33\% as compared to the time-sharing scheme with $M_{\mathrm{t}} = 40$ and $M_{\mathrm{t}} = 30$, respectively. Similar to the average S\&C performance depicted in Fig.~\ref{figure9}, the S-C ASE region for the proposed cooperative ISAC scheme significantly expands compared to the time-sharing scheme as the number of transmit antennas grows. In addition, as shown in Fig. \ref{figure11}, as the number of antennas increases, the communication ASE improves more significantly than the sensing ASE. The main reason is that the adopted beamforming method can achieve multiplexing and  diversity gains for communication.

\begin{figure}[t]
	\centering
	\includegraphics[width=7.5cm]{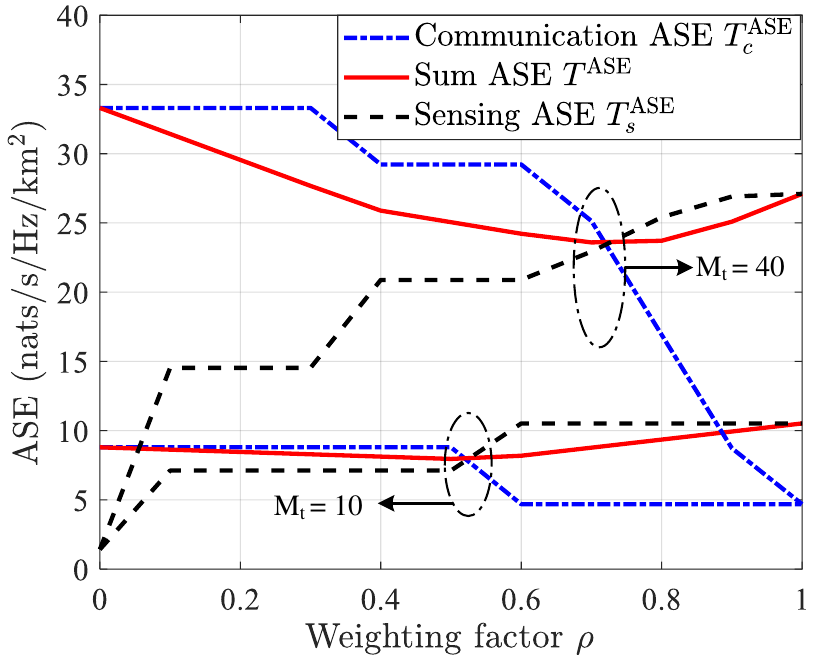}
	\caption{Weighted sum ASE comparison with respect to weighting factor $\rho$.}
	\label{figure12}
\end{figure}

In Fig.~\ref{figure10}, with $M_{\mathrm{t}} = 30$ shows the boundary comparison of network performance under different link-level quality constraints, i.e., $L \ge L^{th}$ and $Q \ge Q^{th}$, where $Q^{th}$ and $L^{th}$ represent the minimum cooperative cluster sizes to achieve the required S\&C quality for the served targets/users. The relationship between the link quality constraint and network performance boundary is evident: as the link quality constraint becomes more stringent, the network performance boundary diminishes. This effect is primarily due to the constraint restricting the optimization range of variables $L$ and $Q$, consequently reducing the DoF in balancing the network performance of S\&C. In addition, as the link quality constraints become more stringent, the performance gain achieved by our proposed cooperative method gradually decreases compared to the time-sharing scheme. This behavior can be explained as follows: The link quality constraint actually guarantees the quality of S\&C at the same time, and therefore limits the feasible solution space of ASE optimization. In turn, compared to the time-sharing scheme, the collaboration gain of our proposed method actually comes from the sacrifice of link quality, which improves the overall performance of scheduled users in the network. Fig.~\ref{figure12} shows the optimal weighted summation ASE $T^{\rm{ASE}}$ for problem (P1) under different weighting factor $\rho$. As the weighting coefficient $\rho$ increases, the sensing/communication ASE increases/decreases monotonically. Moreover, it can be found that $T^{\rm{ASE}}$ is always less than the maximum of the sensing ASE and the communication ASE, since the weighted summation ASE value is always among the sensing-only ASE and the communication-only ASE.

\section{Conclusion}
In this paper, we proposed a novel cooperative scheme in ISAC networks by adopting the coordinated beamforming technique. With SG tools, the S\&C performance in ASE expression is described analytically. We revealed that, unlike communication, the distance PDF of sensing interference is affected by the interference hole region, and we presented a mathematical method to obtain a tight tractable expression by removing the probability with respect to areas where interference does not exist. We proved that it is not necessary to perform interference nulling when maximizing communication ASE; by contrast, interference nulling is required when maximizing sensing ASE. We formulated a profile optimization problem for ISAC network performance, whose performance is compared to the time-sharing scheme to verify that the optimal spatial resource allocation in ISAC networks can effectively improve the cooperative gain at the network level. The simulation results demonstrate the benefits of the proposed cooperative ISAC scheme and provide insightful guidelines for designing the practical large-scale ISAC networks. Performing performance analysis on ISAC networks in low load scenarios and dense clutter interference environments provides valuable avenues for future research. Further investigation is needed into the joint multi-target mutual information analysis, taking into account inter-target dependencies. In addition, sensing performance metrics based on estimation theory, such as the Cramer-Rao lower bound (CRLB), are also worth investigating in network-level ISAC. 

\section*{Appendix A: \textsc{Proof of sensing ASE expression}}

Before deriving the mutual information rate for sensing, we first provide a practical approximated mutual information rate for multi-target sensing, and then derive the approximated mutual information rate for single-target sensing. Let ${\bm{\zeta}}$ denote the target parameters of all targets and ${\bm{\zeta}}_j$ denote the parameters of target $j$. 
		\begin{thm}
			If the reflections from the targets can be approximated as independent, the mutual information rate can be approximated as
			\begin{equation}
				I({\bf{y}}_s;{\bm{\zeta}}|{\bf{X}}_1) \approx \sum\nolimits_{j = 1}^J I({\bf{y}}_s;{\bm{\zeta}}_j|{\bf{X}}_1).
			\end{equation}
			where ${\bm{\zeta}} = \{{\bm{\zeta}}_1,\cdots,{\bm{\zeta}}_J\}$ denotes the parameters of $J$ targets, and ${\bm{\zeta}}_1$ denotes the parameters of target $j$. 
	\end{thm}
	\begin{proof}
		First, the mutual information of multi-target parameters can be written as 
		\begin{equation}
			I({\bf{y}}_s;{\bm{\zeta}}|{\bf{X}}_1) = \sum\nolimits_{j = 1}^J I({\bf{y}}_s;{\bm{\zeta}}_j | {\bm{\zeta}}_{-j},{\bf{X}}_1),
		\end{equation} 
		where ${\bm{\zeta}}_{-j}$ denotes all target parameters except ${\bm{\zeta}}_j$. If the targets' reflections can be treated as approximately independent, the mutual information for each target becomes independent of others. 
		This approximation assumes that the conditional dependence on other targets can be neglected. Therefore, we have
		\begin{equation}
			I({\bf{y}}_s;{\bm{\zeta}}|{\bf{X}}_1) \approx \sum\nolimits_{j = 1}^J I({\bf{y}}_s;{\bm{\zeta}}_j|{\bf{X}}_1).
		\end{equation}
		This thus completes the proof.
	\end{proof}
	
	To efficiently model and optimize network performance, our approach employs the above practical approximation by summing individual terms for each target, served by the BS. 
	In this case, we can analyze the mutual information for all served targets separately. In the following, we simplify the notation by omitting the index $j$. The target's channel matrix $
	{\bf{H}}_s=\xi \cdot {\bf{b}}(\theta) \cdot {\bf{a}}(\theta)^H$. To facilitate the analysis of the sensing mutual information rate, we demonstrate that the target channel can be approximated as following a Gaussian distribution when the variance of $\theta$ is small.
	\begin{thm}\label{SmallVariance}
		With small variance of $\theta$, in ${\bf{H}}_s$, $h_{m,n}$ can be approximated by a Gaussian distribution.
	\end{thm}
	\begin{proof}
		The complex Gaussian random variable $\xi$ is defined as $\xi \sim \mathcal{C N}\left(0, \sigma_{\xi}^2\right)$, where the real and imaginary components are independent and identically distributed as $\mathcal{N}\left(0, \frac{\sigma_{\xi}^2}{2}\right)$.
		First, the element in the $n$th row and the $m$th column, $h_{m, n}=\xi \cdot e^{-j \frac{2 \pi d}{\lambda}(m-1) \sin (\theta)} \cdot e^{j \frac{2 \pi d}{\lambda}(n-1) \sin (\theta)}$. It can be simplified as $h_{m, n}=\xi \cdot e^{-j \frac{2 \pi d}{\lambda}(m-n) \sin (\theta)}$. $\sin (\theta)$ is a nonlinear transformation of the Gaussian variable $\theta$. When $\sigma_\theta^2$ is small, we can make a useful approximation, a Taylor expansion of $\sin (\theta)$ can be given by
		$$
		\sin (\theta) \approx \sin \left(\mu_\theta\right)+\cos \left(\mu_\theta\right)\left(\theta-\mu_\theta\right).
		$$
		Thus, the expression becomes $
		h_{m, n}=\xi e^{-j \frac{2 \pi d}{\lambda}(m-n)\left(\sin \left(\mu_\theta\right)+\cos \left(\mu_\theta\right)\left(\theta-\mu_\theta\right)\right)}$.
		By linearizing the complex exponential, the complex exponential can be expanded as:
		\begin{equation}
			\begin{aligned}
				h_{m, n} \approx & \xi e^{-j \frac{2 \pi d}{\lambda}(m-n) \sin \left(\mu_\theta\right)} \\
				&\times \left(1-j \frac{2 \pi d}{\lambda}(m-n) \cos \left(\mu_\theta\right)\left(\theta-\mu_\theta\right)\right).
			\end{aligned}
		\end{equation}
		Here, $\theta-\mu_\theta$ is a Gaussian variable with zero mean and variance $\sigma_\theta^2$. Since the term $\left(\theta-\mu_\theta\right)$ is multiplied by a small constant, its effect is limited. Then, the distribution of the element $h_{m, n}$ can be approximated by:
		$$
		h_{m, n} \sim \mathcal{C N}\left(0, \sigma_{\xi}^2\left(1+\left(\frac{2 \pi d}{\lambda}\right)^2(m-n)^2 \cos ^2\left(\mu_\theta\right) \sigma_\theta^2\right)\right).
		$$
		Based on the above discussion, we have
		$$
		\mathbf{H}_s \approx \xi \left(\mathbf{A}_0+\left(\theta-\theta_0\right) \mathbf{B}_0\right),
		$$
		where $\mathbf{A}_0=\mathbf{b}\left(\theta_0\right) \mathbf{a}^H\left(\theta_0\right)$
		and $\mathbf{B}_0=\mathbf{b}^{\prime}\left(\theta_0\right) \mathbf{a}^H\left(\theta_0\right)+\mathbf{b}\left(\theta_0\right) \mathbf{a}^{\prime H}\left(\theta_0\right)$.
		This thus completes the proof.
	\end{proof}
	
Lemma \ref{SmallVariance} is practical because the accuracy of angle estimation is typically high, which supports the assumption that the variance of $\sigma^2_{\theta}$ is small and aligns with common practices in the field. Furthermore, in practical scenarios, the target moves at a limited speed, causing only small angle changes over short time intervals.
	Based on the approximation in Lemma 2, we assume $h_{m,n}$ can be approximated by a Gaussian distribution. To simplify the derivation, we utilize the original expression of ${\bf{H}}_s$ to derive the mutual information.
	Then, the mutual information $I({\mathbf{y}}_s ; \mathbf{H}_s \mid {\mathbf{X}}_1)$ is given by:
	\begin{equation}\label{MutualInformation}
		I({\mathbf{y}}_s ; \mathbf{H}_s \mid {\mathbf{X}}_1)=h({\mathbf{y}}_s \mid {\mathbf{X}}_1)-h({\mathbf{y}}_s \mid \mathbf{H}_s, {\mathbf{X}}_1),
	\end{equation}
	where $h({\mathbf{y}}_s \mid {\mathbf{X}}_1)$ is the differential entropy of ${\mathbf{y}}_s$ given ${\mathbf{X}}_1$.
	$h({\mathbf{y}}_s \mid \mathbf{H}_s, {\mathbf{X}}_1)$ is the conditional differential entropy of ${\mathbf{y}}_s$ given both $\mathbf{H}_s$ and ${\mathbf{X}}_1$. In the following, we simplify the notation by omitting $\theta_j$.
	The conditional differential entropy is given by: 
	\begin{equation}
		\begin{aligned}
			&h({\mathbf{y}}_s \mid \mathbf{H}_s, {\mathbf{X}}_1) \\
			= & \log \left( \det \left(\pi e \sum\nolimits_{i = Q+1}^{\infty}{d_i^{-\alpha} \bf{h}}_{q,1}^H \mathbf{R}_q {\bf{h}}_{q,1} + \sigma_n^2 \|\mathbf{v}\|^2\right)\right),
		\end{aligned}
	\end{equation}
	where ${\bf{h}}_{q,1}^H = {\bf{v}}^H {\bf{G}}_{q,1}^H$, and $\mathbf{R}_{q}$ denotes the covariance matrix of BS $q$'s  transmit signals. Then, we have
	\begin{equation}\label{FunctionHY}
		h({\mathbf{y}}_s | {\mathbf{X}}_1) = \log\left(\pi e   \mathbf{R}_y \right),
	\end{equation}
where 
\begin{equation}
	\begin{aligned}
		&\mathbf{R}_y=\mathbb{E}\left[{{\mathbf{y}}_s {\mathbf{y}}^H_s}\right] \\
		&=d_1^{-2\beta}\mathbf{v}^H \mathbf{H} {\bf{R}}_1 \mathbf{H}^H \mathbf{v} + \sum\nolimits_{i = Q+1}^{\infty}d_i^{-\alpha}{\bf{h}}_{q,1}^H {\mathbf{R}}_q {\bf{h}}_{q,1}^H +\sigma_n^2 \|\mathbf{v}\|^2.
	\end{aligned}
\end{equation} 
	Similarly, the $h({\mathbf{y}}_s | {\bf{H}}_s,{\mathbf{X}}_1)$ can be given by
	\begin{equation}\label{FunctionyHX1}
		h({\mathbf{y}}_s | {\bf{H}}_s,{\bf{X}}_1) \! =\! \log\!\left(\!\pi e   \! \!\sum\nolimits_{i = Q+1}^{\infty} \!d_i^{-\alpha}{\bf{h}}_{q,1}^H {\mathbf{R}}_q {\bf{h}}_{q,1} + \sigma_n^2 \|\mathbf{v}\|^2\!\right)\!.
	\end{equation}
	Plugging (\ref{FunctionHY}) and (\ref{FunctionyHX1}) into (\ref{MutualInformation}), we have
	\begin{equation}
		\begin{aligned}
			I({\mathbf{y}}_s; \mathbf{H} \!\mid\! {\mathbf{X}}_1)
			\!=  \!\log \!\left(\!1  \!+  \!\frac{ \Delta T |\xi|^2 d_i^{-2\beta}({\bf{v}}^H {\bf{b}}(\theta))^2  ({\bf{a}}^H(\theta) {\mathbf{X}}_1) ^2}{\sum\nolimits_{i = Q+1}^{\infty}d_i^{-\alpha}{\bf{h}}_{q,1}^H \mathbf{R}_q {\bf{h}}_{q,1} + \sigma_n^2 \|\mathbf{v}\|^2} \!\right)\!.
		\end{aligned}
	\end{equation}
	In the context of significant interference in dense cell scenarios, this paper specifically addresses an interference-limited network where noise is considered negligible. Therefore, the mutual information rate can be further approximated as
	\begin{equation}
		\begin{aligned}
			I({\mathbf{y}}_s; \mathbf{H} \!\mid \! {\mathbf{X}}_1)
			&\!=\!  \log \!\left(\! 1 + \! \frac{\Delta T |\xi|^2d_i^{-2\beta}({\bf{v}}^H {\bf{b}}(\theta))^2  ({\bf{a}}^H(\theta) {\mathbf{X}}_1) ^2}{\sum\nolimits_{i = Q+1}^{\infty}d_i^{-\alpha}{\bf{h}}_{q,1}^H \mathbf{R}_q {\bf{h}}_{q,1}}  \!\right) \\
			&=  \log \left(1 + {\rm{SIR}}_s \right),
		\end{aligned}
\end{equation}
where ${\rm{SIR}}_s =  \Delta T \kappa M_{\mathrm{r}} |\xi|^2 \frac{{h_{j,1}^t{{\left\| {{{\bf{d}}_1}} \right\|}^{ - 2\beta }}}}{{\sum\nolimits_{q = Q+1}^\infty  {{h_{q,1}}} {{\left\| {{{\bf{d}}_q} - {{\bf{d}}_1}} \right\|}^{ - \alpha }}}}$.

\section*{Appendix B: \textsc{Proof of Theorem \ref{CommunicationTightExpression}}}
The interference term with a given distance $r$ from the typical user to the serving BS, can be derived by utilizing Laplace transform. For ease of analysis, we introduce a geometric parameter $\eta_L = \frac{\left\| {{{\bf{d}}_1}} \right\|}{\left\| {{{\bf{d}}_{L}}} \right\|} $, defined as the distance to the closest BS normalized by the distance to the furthest BS in the cluster for typical user. When $\left\| {{{\bf{d}}_1}} \right\| = r$ and $\left\| {{{\bf{d}}_L}} \right\| = r_L$, we have
\begin{equation}\label{CommunicationEquationExpression}
	\begin{aligned}
		&{{\cal L}_{{I_{\rm{C}}}}}(z) = {\rm{E}}_{g, \Phi_b}\left[ {\exp \left( { - z{{\left\| {{{\bf{d}}_1}} \right\|}^\alpha }\sum\nolimits_{i = L+1}^\infty  {{{\left\| {{{\bf{d}}_i}} \right\|}^{ - \alpha }}} {{| {{\bf{h}}_{k,i}^H{{\bf{W}}_i}} |}^2}} \right)} \right]  \\
		&\overset{(a)}{=}  {\rm{E}}_{\Phi_b}\left[ \left( \prod _{{{{\bf{d}}_i}} \in \Phi_b \textbackslash {\cal{O}}(0,r_L)} {   {  {{{{\left( {1 + z{r^\alpha }{{\left\| {{{\bf{d}}_i}} \right\|}^{ - \alpha }}} \right)}^{-K}}}} } dx} \right) \bigg| r, r_L \right] \\
		&\overset{(b)}{=} \exp \left( { - 2\pi \lambda_b \int_{{r_L}}^\infty  {\left( {1 - {{{{\left( {1 + z{r^\alpha }{x^{ - \alpha }}} \right)}^{-K}}}}} \right)} xdx} \right) \\
		&\overset{(c)}{=} \exp \bigg(  - \pi \lambda_b {r^2}{\rm{H}}\left( {z,K,\alpha ,\eta_{L} } \right) \bigg) ,
	\end{aligned}
\end{equation}
where ${\rm{H}}\left( {z,K,\alpha ,\eta_L } \right) = K{z^{\frac{2}{\alpha }}}B\left( {\frac{z}{{z + {\eta_L ^{ - \alpha }}}},1 - \frac{2}{\alpha },K + \frac{2}{\alpha }} \right)+ \frac{1}{{{\eta_L ^2}}}\left( {\frac{1}{{{{\left( {1 + z{\eta_L ^\alpha }} \right)}^K}}} - 1} \right)$.
In (\ref{CommunicationEquationExpression}), ($a$) follows from the fact that the small-scale channel fading is independent of the BS locations and that the interference power imposed by each interfering BS at user $k$ is distributed as $\Gamma(K,1)$. Then, relation ($a$) is then obtained using the moment generating function (MGF) of this Gamma distribution. To derive ($b$), we use the probability generating functional (PGFL) of a PPP with density $\lambda_b$. ($c$) follows from distribution integral strategies and $\eta_{L} = \frac{r}{r_L}$. 

Then, we perform the marginal probability integral over $r$. By plugging (\ref{CommunicationEquationExpression}) into the following equation, where ${f_r}\left( r \right) = 2\pi {\lambda _b}r{e^{ - \pi {\lambda _b}{r^2}}}$, it follows that
\begin{equation}\label{LapaceformCOmmunication}
	\begin{aligned}
		{{\cal L}_{{I_{\rm{C}}}}}(z) = & \int_0^\infty  {\exp \left( { - \pi \lambda {r^2}{\rm{H}}\left( {z,K,\alpha ,\eta_{L} } \right)} \right)} f(r)dr \\
		= & \frac{1}{{{\rm{H}}\left( {z,K,\alpha ,\eta_{L} } \right) + 1}}.
	\end{aligned}
\end{equation}
According to Lemma 3 in \cite{Zhang2014StochasticGeometry}, the PDF of the distance ratio $\eta_L$ can be given by 
\begin{equation}\label{RatioEquation}
	{f_{\eta_L}}\left( x \right) = 2\left( {L - 1} \right)x{\left( {1 - {x^2}} \right)^{L - 2}}.
\end{equation}
Then, we perform the marginal probability integral over $\eta_L$. By plugging (\ref{LapaceformCOmmunication}) and (\ref{RatioEquation}) into (\ref{ASEcommunication}), the achievable rate can be denoted by
\begin{equation}
	\begin{aligned}
		R_c =& \int_0^\infty  {\frac{{1 - {{\left( {1 + z} \right)}^{ - {M_{\mathrm{t}} - KL - J(Q-1) + 1}}}}}{z}} \\
		& \int_0^1 {\frac{{2\left( {L - 1} \right)\eta_L {{\left( {1 - {\eta_L ^2}} \right)}^{L - 2}}}}{ {\rm{H}}\left( {z,K,\alpha ,\eta_L } \right) +1 }d \eta_L }dz.
	\end{aligned}
\end{equation}
This thus completes the proof.

\section*{Appendix C: \textsc{Proof of Theorem \ref{CommunicationLooseExpression}}}

We first derive the expression without interference nulling, and then obtain an approximate expression based on MISR method. Specifically, under the given distance from the typical user to the serving BS $r$, the conditional expectation can be further approximated as follows:
\begin{equation}\label{ConditionalExpectionC}
	\begin{aligned}
		&{\rm{E}}_{g, \Phi_b}\left[ {\log \left( {1 + \mathrm{SIR}_c} \right)} \big| r\right] \\
		\approx& {\rm{E}}_{\Phi_b}\left[ {\log \left( {1 + \frac{{\rm{E}}\left[{g_{k,1}^k}\right]}{{\sum\nolimits_{{{i}} = L+1}^\infty  {\rm{E}}\left[{{g_{k,i}}}\right] {{\left\| {{{\bf{d}}_{{i}}}} \right\|}^{ - \alpha }}{r^\alpha }}}} \right)} \right] \\
		=& \int_0^\infty  {\frac{{1 - {{e^{ - z\frac{{M_{\mathrm{t}} - KL - J(Q-1) + 1}}{K}}}} }}{ z }} {\rm{E}}\left[ {{e^{ - z I_{\rm{C}}}}} \right]dz.
	\end{aligned}
\end{equation}
Then, the Laplace transform of communication interference with $L = 1$ can be given by
\begin{equation}
	\begin{aligned}
	{\rm{E}}\left[ {{e^{ - z  I_{\rm{C}}}}} \right]  =& 	\exp \bigg(  - \left. {\pi \lambda x\left( {1 - {e^{ - z{r^\alpha }{x^{ - \alpha /2}}}}} \right)} \right|_{{r^2}}^\infty  \\
		&- \pi \lambda \int_{{r^2}}^\infty  {\frac{\alpha }{2}} z{r^\alpha }{x^{ - \frac{\alpha }{2}}}{e^{ - z{r^\alpha }{x^{ - \alpha /2}}}}dx \bigg)\\
		=& \exp \left( - {\pi \lambda {r^2}\left({\rm{F}}\left( {z,\alpha} \right) - 1\right)} \right),
	\end{aligned}
\end{equation}
where ${\rm{F}}\left( {z,\alpha} \right) = {{e^{ - z}}  + {z^{\frac{2}{\alpha } }}\Gamma \left( {1 - \frac{2}{\alpha },z} \right)}$.  By plugging ${\rm{E}}\left[ {{e^{ - z  I_{\rm{C}}}}} \right]$ and (\ref{ConditionalExpectionC}) into $\int_0^\infty {\rm{E}}_{g, \Phi_b}\left[ {\log \left( {1 + \mathrm{SIR}_c} \right)} \big|r\right]{f_r}\left( r \right) dr$, we have ${\rm{E}}_{g, \Phi_b}\left[ {\log \left( {1 + \mathrm{SIR}_c} \right)} \right] = \int_0^\infty  {\exp \left( { - \pi \lambda {r^2}({\rm{F}}\left( {z,\alpha } \right)-1)} \right)} f(r)dr = \frac{1}{{{\rm{F}}\left( {z,\alpha} \right) }}$.
Then, we need to further obtain the  effective gain of SIR with $L$ cooperative BSs compared to the case with $L = 1$. According to \cite{Haenggi2014MISR}, 
we can obtain an approximate coverage probability with $L-1$ interference nulling based on no-interference nulling case, i.e., $P_{L}(\gamma) = P_1(\gamma/G_L)$, where $G_L = \frac{{\Gamma \left( {L + \frac{\alpha }{2}} \right)}}{{\Gamma \left( {L + 1} \right)\Gamma \left( {1 + \frac{\alpha }{2}} \right)}}$ and $P_L(\gamma)$ denotes the coverage probability with SIR threshold $\gamma$. Therefore, $\frac{{\rm{E}}\left[{g_{k,1}^k}\right] }{{\sum\nolimits_{{{i}} = L+1}^\infty  {\rm{E}}\left[{{g_{k,i}}}\right] {{\left\| {{{\bf{d}}_{{i}}}} \right\|}^{ - \alpha }}{r^\alpha }}}$ and $\frac{{\rm{E}}\left[{g_{k,1}^k}\right] G_L}{{\sum\nolimits_{{{i}} = 2}^\infty  {\rm{E}}\left[{{g_{k,i}}}\right] {{\left\| {{{\bf{d}}_{{i}}}} \right\|}^{ - \alpha }}{r^\alpha }}}$ follow nearly identical distributions. Then, it follows that
\begin{equation}
	\begin{aligned}
		R_c \approx & \int_0^\infty  {\frac{{\left( {1 - {e^{ - z G_L  {\frac{{M_{\mathrm{t}} - KL - J(Q-1) + 1}}{K}} }}} \right)}}{z {{\rm{F}}\left( {z,\alpha} \right) }}}  dz \\ 
		= & \int_0^\infty  {\frac{{ {1 - {e^{ - z {\frac{{\Gamma \left( {L + \frac{\alpha }{2}} \right)}\left( {{M_{\mathrm{t}}-J(Q-1)+1} - KL} \right)}{K {\Gamma \left( {L + 1} \right)\Gamma \left( {1 + \frac{\alpha }{2}} \right)}}} }}} }}{{z {\rm{F}}(z,\alpha)}}} dz .
	\end{aligned}
\end{equation}
This thus completes the proof.

\section*{Appendix D: \textsc{Proof of Proposition \ref{NoCooperativeNecessary}}}
When $\alpha = 4$, $G_L = \frac{L+1}{2}$. For any given $K$, $Q$, $J$, we have ${\frac{{\Gamma \left( {L + \frac{\alpha }{2}} \right)}}{{\Gamma \left( {L + 1} \right)\Gamma \left( {1 + \frac{\alpha }{2}} \right)}}}{\frac{{{{M_{\mathrm{t}} -KL - J(Q-1) + 1}}}}{K}}  = \frac{1}{2}\left( { - {L^2} + \left( {\frac{{M_{\mathrm{t}} - J(Q-1) + 1}}{K} - 1} \right)L + \frac{{M_{\mathrm{t}} - J(Q-1) + 1}}{K}} \right) \triangleq y(L) $. $y(L)$ is a quadratic function about $L$. It can be found that when $\frac{{M_{\mathrm{t}} - J(Q-1) - K + 1}}{{2K}} \le 1$, $y(L)$ is maximized when $L = 1$. Otherwise, $y(L)$ is maximized when $L = \frac{{M_{\mathrm{t}} - J(Q-1) - K + 1}}{{2K}}$, let $v = \frac{K}{{M_{\mathrm{t}} - J(Q-1) + 1}}$, we have
\begin{equation}
	\begin{aligned}
		\frac{\partial T^{\rm{ASE}}_c}{\partial v}  &= \int_0^\infty  {\frac{{1 - {e^{ - z\frac{1}{4}{{\left( {\frac{1}{v} + 1} \right)}^2}}} - z\frac{1}{2}\left( {\frac{{1 + v}}{{{v^2}}}} \right){e^{ - z\frac{1}{4}{{\left( {\frac{1}{v} + 1} \right)}^2}}}}}{{zF(z,\alpha )}}} dz \\
		&= \int_0^\infty  {\frac{{1 - \left( {1 + z\frac{{1 + v}}{{2{v^2}}}} \right){e^{ - z\frac{1}{4}{{\left( {\frac{1}{v} + 1} \right)}^2}}}}}{{zF(z,\alpha )}}} dz \ge 0.
	\end{aligned}
\end{equation}
Thus, in this case with $\frac{{M_{\mathrm{t}} - J(Q-1) - K + 1}}{{2K}} \ge 1$, $T^{\rm{ASE}}_c$ is a monotonically increasing function with respect to $K$. Since $K \le \frac{{M_{\mathrm{t}}  - J(Q-1) + 1}}{3}$, $y(L)$ is maximized with $K = \frac{{M_{\mathrm{t}}  - J(Q-1) + 1}}{3}$, and in this case, $L = \frac{{M_{\mathrm{t}}  - J(Q-1) + 1}}{{2K}} - \frac{1}{2} = 1$. Combining the above analysis, we complete the proof.

\section*{Appendix E: \textsc{Proof of Proposition \ref{LaplaceTransformSensing}}}

As shown in Fig. \ref{figure3}, to accurately derive the PDF for the distance, we must determine the count of points distributed outside the hole within the strip. The strip's area can be represented as $2 x \mathrm{~d} x(\pi-\varphi(x))$, with $\varphi(x) = \arccos \left(\frac{x}{2R}\right)$ denoting the angle within the hole, as depicted in Fig. \ref{figure3}. Consequently, the Poisson distribution models the number of interfering points within this strip, with a mean of $\lambda_b 2 x \mathrm{~d} x(\pi-\varphi(x))$. Given that the precise point locations within the strip are inconsequential, we can disregard the hole and redistribute the points uniformly across the entire strip, which has an area of $2 \pi x \mathrm{~d} x$. This implies that a PPP within the hole can be effectively represented as a non-homogeneous PPP with a density of $\lambda_b(1-\frac{\varphi(x)}{\pi})$. Using this intermediate result and the above insights, we now derive tight bounds on the Laplace transform of sensing interference as follows:
\begin{equation}\label{HoleExpression}
	\begin{aligned}
		&	{{\cal L}_{{I_{\rm{S}}}}}(z) =  \exp \bigg(  - 2\pi \lambda_b \int_0^\infty  {\bigg( {1 - \frac{1}{{{{\left( {1 + z{R^{2\alpha }}{x^{ - \beta }}} \right)}^K}}}} \bigg)} xdx  \\
		& + \lambda_b \int_0^{2R} {\underbrace {2\arccos  {\frac{x}{{2R}}} }_{\text{{{{angle}}  {{in}} {{hole}}}}}\bigg( {1 - \frac{1}{{{{\left( {1 + z{R^{2\alpha }}{x^{ - \beta }}} \right)}^K}}}} \bigg)xdx}  \bigg) \\
		& \mathop  = \limits^{\frac{x}{R} = t}  \exp \bigg( { - \pi \lambda_b K{z^{\frac{2}{\beta }}}{R^{\frac{{4\alpha }}{\beta }}}B\left( {1,1 - \frac{2}{\beta },K + \frac{2}{\beta }} \right)} \\
		&+  {\lambda_b {R^2}\int_0^2 {2\arccos  {\frac{t}{2}}\bigg( {1 - \frac{1}{{{{\left( {1 + z{R^{2\alpha  - \beta }}{t^{ - \beta }}} \right)}^K}}}} \bigg)} tdt} \bigg).
	\end{aligned}
\end{equation}
In (\ref{HoleExpression}), the second part is to subtract the interference in the hole.
By calculating the probability integral of the target distance $R$, we have
\begin{equation}\label{LapalaceTransform}
	\begin{aligned}
		&{{\cal L}_{{I_{\rm{S}}}}}(z) = \\
		&\int_0^\infty \! \exp \bigg( \! - R \bigg( K{z^{\frac{2}{\beta }}}{{\left( {\frac{R}{{\pi \lambda_b }}} \right)}^{\frac{{2\alpha }}{\beta } - 1}} \! B\left( \! {1,1 - \frac{2}{\beta },K + \frac{2}{\beta }} \right) \!+\! 1 \\
		&-\! \int_0^2 \! {\frac{2}{\pi }\arccos  {\frac{t}{2}} \bigg( {1 \! - {{{{( {1 + z{{( {\frac{R}{{\pi \lambda_b }}} )}^{\alpha  - \beta /2}}{t^{ - \beta }}} )}^{-K}}}}} \bigg)} tdt \bigg) \bigg) dR.
	\end{aligned}
\end{equation}
Then, by plugging (\ref{LapalaceTransform}) into (\ref{RadarRateExpression}), it follows that
\begin{equation}
	R_s = \int_0^\infty  {\frac{{1 - {{{\left( {1 +  \Delta T \kappa M_{\mathrm{r}} |\xi|^2 z} \right)}^{ - K}}}}}{z}} \int_0^\infty {{\cal L}_{{I_{\rm{S}}}}}(z) f(R) dR dz.
\end{equation}
This completes the proof.

\section*{Appendix F: \textsc{Proof of Theorem \ref{ASEsensingExpression}}}
First, we draw the following conclusion to facilitate the derivation of sensing ASE expression.
\begin{thm}\label{IgnorableHole}
	When $Q \gg 1$, the term $\exp \left( {\lambda_b {R^2}\int_{{r_Q}/{R}}^2 {2\arccos \left( {\frac{t}{2}} \right)\left( {1 - \frac{1}{{{{\left( {1 + z{R^{2\alpha  - \beta }}{t^{ - \beta }}} \right)}^K}}}} \right)} tdt} \right) \to 0$ in (\ref{HoleExpression}).
\end{thm}
\begin{proof}
	First, we ignore the interference in the hole (c.f. Fig.~\ref{figure3}). 
	In this case, it is worth noting that the PDF of the distance from the $(Q-1)$th closest BS to the serving BS is the same as that from the $Q$th BS to the origin. This can be proved by Slivnyak's theorem \cite{mukherjee2014analytical}, i.e., for a PPP, because of the independence between all of the points, conditioning on a point at $x$ does not change the distribution of the rest of the process. The distance from the $(Q-1)$th closest BS to the serving BS is denoted by $r_{Q}$. Therefore, we have the PDF of the distance $r_Q$ from the $(Q-1)$th closest BS to the serving BS can be given by ${f_{{r_{Q}}}}\left( r \right) = {e^{ - \lambda \pi {r^2}}}\frac{{2{{\left( {\lambda \pi {r^2}} \right)}^Q}}}{{r\Gamma \left( Q \right)}}$, and according to \cite{Zhang2014StochasticGeometry}, the complementary cumulative distribution function (CCDF) of $\frac{r_Q}{2R}$ can be expressed as $P\left[\frac{r_Q}{2R} \ge x\right] = 1 - \left(1 - \frac{1}{4x^2}\right)^{Q-2}$.
	It can be readily proved that when $Q \gg 1$,  $P[\frac{r_Q}{2R} \ge x] \to 1$ for any given $x > 1$, which represents the interference distance from the BS at ${\bf{d}}_Q$ to the serving BS is always larger than $2R$. Therefore, the probability of interference in the hole is almost zero due to cooperative interference nulling when $Q \gg 1$, i.e., $\frac{r_Q}{R} \ge 2$. This thus proves that $\exp \left( {\lambda_b {R^2}\int_{{r_Q}/{R}}^2 {2\arccos \left( {\frac{t}{2}} \right)\left( {1 - \frac{1}{{{{\left( {1 + z{R^{2\alpha  - \beta }}{t^{ - \beta }}} \right)}^K}}}} \right)} tdt} \right) \to 0$ when $Q \gg 1$.
\end{proof}

According to Lemma \ref{IgnorableHole}, the interference term in the hole can be ignored when $Q \gg 1$. Thus, to simplify the expression derivation, we remove the second term in (\ref{SensingRateExpression}) for the case with $Q \ge 2$, and then the conditional Laplace transform of sensing interference can be given by 
\begin{equation}
	\begin{aligned}
		&\left[ {{e^{ - z I_{\rm{S}}}}} \big|R,r_Q \right] = \exp \bigg(  - \pi \lambda_b \bigg( r_Q^2\left( {{{{{\left( {1 + z{R^{2\alpha }}r_Q^{ - \beta }} \right)}^{-K}}}} - 1} \right) \\
		&+ K{z^{\frac{2}{\beta }}}{R^{\frac{{4\alpha }}{\beta }}}B\bigg( {\frac{{z{R^{2\alpha }}r_Q^{ - \beta }}}{{z{R^{2\alpha }}r_Q^{ - \beta } + 1}},1 - \frac{2}{\beta },K + \frac{2}{\beta }} \bigg) \bigg) \bigg).
	\end{aligned}
\end{equation}
Then, the Laplace transform of sensing interference is $\left[ {{e^{ - z I_{\rm{S}}}}} \right] = \int_0^\infty \int_0^\infty \left[ {{e^{ - z I_{\rm{S}}}}} \big|R,r_Q \right] {f_{{r_{Q}}}}\left( r_Q \right)f_r(R) dRdr_Q$, where $f_r(R)= 2\pi {\lambda _b}R{e^{ - \pi {\lambda _b}{R^2}}}$. Finally, by plugging $\left[ {{e^{ - z I_{\rm{S}}}}} \right]$ into (\ref{RadarRateExpression}) and (\ref{ASEsensing}), the tractable expressions of the radar information rate and sensing ASE can be obtained.

\footnotesize  	
\bibliography{mybibfile}
\bibliographystyle{IEEEtran}

\end{document}